\documentclass{article}
\usepackage[margin=0.8in]{geometry}
\usepackage{mathtools}
\usepackage{amssymb}
\usepackage{amsthm}
\usepackage{enumitem}
\usepackage{bigstrut}
\usepackage{arydshln}
\usepackage{subfig}
\usepackage{color}
\usepackage{graphicx}
\usepackage{tikz-cd}
\usepackage[
colorlinks=true,
citecolor=blue,
linkcolor=blue,
]{hyperref}
\usetikzlibrary{automata, positioning,arrows,shapes,decorations.pathmorphing}
\newtheorem{conjecture}{Conjecture}[section]
\newtheorem{theorem}{Theorem}[section]
\newtheorem{proposition}{Proposition}[section]
\newtheorem{lemma}{Lemma}[section]
\newtheorem{corollary}{Corollary}[section]
\theoremstyle{definition}
\newtheorem{example}{Example}[section]
\newtheorem{definition}{Definition}[section]

\usepackage{subfig}
\usepackage{tikz}
\usepackage{bigstrut}
\usepackage{arydshln}

\newcommand{\myparagraph}[1]{{\textbf{#1}.}}

\newcommand{\db}{\mathbf{db}}

\newcommand{\FO}{$\mathbf{FO}$}
\newcommand{\PTIME}{$\mathbf{PTIME}$}
\newcommand{\LSPACE}{$\mathbf{L}$}

\newcommand{\coNP}{$\mathbf{coNP}$}

\newcommand{\NL}{$\mathbf{NL}$}
\newcommand{\LFPL}{$\mathbf{LFP}$}

\newcommand{\sjf}[1]{{#1}^{\mathsf{sjf}}}
\newcommand{\attacks}[1]{\stackrel{#1}{\leadsto}}

\newenvironment{jefenv}{\mbox{}\newline\color{olive}}{}

\newtheorem{claim}{Claim}[section]

\newcommand{\defeq}{\vcentcolon=}

\newcommand{\formula}[1]{\left({#1}\right)}
\newcommand{\card}[1]{|{#1}|}

\newcommand{\adom}[1]{{\mathsf{adom}}({#1})}
\newcommand{\cqa}[1]{\mathsf{CERTAINTY}({#1})}
\newcommand{\certaintrace}[1]{\mathsf{CERTAIN}_{\mathsf{tr}}({#1})}
\newcommand{\rootedAt}[2]{{#2}^{{#1}}_{{\vartriangle}}}

\newcommand{\queryvars}[1]{{\mathbf{vars}}({#1})}
\newcommand{\leqhomo}[2]{\leq_{{#1}\rightarrow{#2}}}

\newcommand{\qcopy}[3]{\langle{#1}, {[#2 \rightarrow #3]}\rangle}

\newcommand{\atoms}[1]{\mathsf{atoms}({#1})}

\newcommand{\emptyword}{\varepsilon}

\newcommand{\cbranch}{\mathsf{C}_{\mathsf{branch}}}

\newcommand{\cone}{\mathsf{C}_{1}}
\newcommand{\ctwo}{\mathsf{C}_{2}}

\newcommand{\cfactor}{\mathsf{C}_{\mathsf{factor}}}
\newcommand{\cprefix}{\mathsf{C}_{\mathsf{prefix}}}

\newcommand{\starttwo}[2]{\mathsf{start}({#1}, {#2})}

\newcommand{\sset}[3]{{\mathsf{FrugalSet}}_{#3}({#1},{#2})}
\newcommand{\treecfg}[1]{{\mathsf{CFG^{\clubsuit}}}({#1})}
\newcommand{\streecfg}[2]{{\mathsf{S\text{-}CFG^{\clubsuit}}}({#1},{#2})}

\newcommand{\derive}{\stackrel{*}{\rightarrow}}

\newcommand{\vardetermines}[1]{\rightarrow_{{#1}}}
\newcommand{\vardisjoint}[1]{\parallel_{{#1}}}
\newcommand{\vars}[1]{\mathsf{vars}({#1})}

\newcommand{\constants}[1]{\mathsf{const}({#1})}

\newcommand{\rep}{\mathbf{r}}

\newcommand{\normalformula}[1]{({#1})}
\newcommand{\pair}[2]{\langle{#1},{#2}\rangle}

\newcommand{\atomAt}[3]{{#1}[{#2}]}
\newcommand{\rewind}[4]{{#1}^{{#4}:{#2}\looparrowright{#3}}}

\newcommand{\precedes}[1]{<_{{#1}}}
\newcommand{\diffbranch}[1]{\parallel_{{#1}}}

\newcommand{\mytotalorder}[1]{\preceq_{#1}}

\newcommand{\factformulapred}{\mathsf{fact}}
\newcommand{\factformula}[4]{\factformulapred({#1}(\underline{{#2}},{#3}),{#4})}

\newcommand{\bcq}{\mathsf{BCQ}}
\newcommand{\ibcq}{\mathsf{GraphBCQ}}
\newcommand{\aibcq}{\mathsf{Graph_{Berge}BCQ}}
\newcommand{\rtbcq}{\mathsf{TreeBCQ}}

\newcommand{\qgraph}[1]{{\mathcal{G}}({#1})}

\usepackage{authblk}
\author[1]{Paraschos Koutris}
\author[1]{Xiating Ouyang}
\author[2]{Jef Wijsen}
\affil[1]{University of Wisconsin-Madison, USA}
\affil[2]{University of Mons, Belgium}
\date{}

\title{Consistent Query Answering for Primary Keys on Rooted Tree Queries}

\begin{document}
\maketitle
 
\begin{abstract}
We study the data complexity of consistent query answering (CQA) on databases that may violate the primary key constraints. 
A repair is a maximal subset of the database satisfying the primary key constraints. 
For a Boolean query~$q$, the problem $\cqa{q}$ takes a database as input, and asks whether or not each repair satisfies~$q$. 
The computational complexity of $\cqa{q}$ has been established whenever $q$ is a self-join-free Boolean conjunctive query, or a (not necessarily self-join-free) Boolean path query.
In this paper, we take one more step towards a general classification for all Boolean conjunctive queries by considering the class of rooted tree queries. 
In particular, we show that for every rooted tree query~$q$, $\cqa{q}$ is in \FO, \mbox{\NL-hard $\cap$ \LFPL}, or \coNP-complete, and it is decidable (in polynomial time), given~$q$, which of the three cases applies. 
We also extend our classification to larger classes of queries with simple primary keys.
Our classification criteria rely on query homomorphisms and our polynomial-time fixpoint algorithm is based on a novel use of context-free grammar (CFG). 
\end{abstract}

\section{Introduction}\label{sec:introduction}









A relational database is \emph{inconsistent} if it violates one or more integrity constraints that are supposed to be satisfied. 
Database inconsistency is a common issue when integrating datasets from heterogeneous sources.
In this paper, we focus on what are probably the most commonly imposed integrity constraints on relational databases: primary keys.
A primary key constraint enforces that no two distinct tuples in the same relation agree on all primary key attributes.

A \emph{repair} of such an inconsistent database instance is naturally defined as a maximal consistent subinstance of the database.
Two approaches are then possible.
In \emph{data cleaning}, the objective is to single out the ``best'' repair, which however may not be practically possible.
In \emph{consistent query answering} (CQA)~\cite{10.1145/303976.303983}, instead of cleaning the inconsistent database instance, we attempt to query \emph{every} possible repair of the database and obtain the \emph{consistent} (or \emph{certain}) answers that are returned across all repairs.
In computational complexity studies, consistent query answering is commonly defined as the following decision problem, for a fixed Boolean query~$q$ and fixed primary keys for all relation names occurring in~$q$:  

\begin{itemize}
	\item[] \textbf{PROBLEM} $\cqa{q}$
	\item[] \textbf{Input}: A database instance $\db$.
	\item[] \textbf{Question}: Does $q$ evaluate to true on every repair of $\db$?
\end{itemize}
The CQA problem for queries $q(\vec{x})$ with free variables is to find all sequences of constants $\vec{c}$, of the same length as~$\vec{x}$, such that $q(\vec{c})$ is true in every repair. We often do not need separate treatment for different constants, in which case we can handle $q(\vec{x})$ as Boolean by treating free variables as if they were constants~\cite{DBLP:journals/pacmmod/FanKOW23,10.1145/3340531.3411911}.

The problem $\cqa{q}$ is obviously in \coNP\ for every Boolean first-order query~$q$. 
It has been extensively studied for $q$ in the class of Boolean conjunctive queries, denoted $\bcq$.
Despite significant research efforts (see Section~\ref{sec:related}), the following dichotomy conjecture remains notoriously open.

\begin{conjecture}\label{conj:dichotomy}
For every query $q$ in $\bcq$, $\cqa{q}$ is 
either in \PTIME\ or \coNP-complete.
\end{conjecture}

An ever stronger conjecture is that the dichotomy of Conjecture~\ref{conj:dichotomy} extends to unions of conjunctive queries.
Fontaine~\cite{10.1145/2699912} showed that this stronger conjecture implies the dichotomy theorem for conservative \emph{Constraint Satisfaction Problems} (CSP)~\cite{10.1145/1970398.1970400,DBLP:journals/jacm/Zhuk20}.

On the other hand, for self-join-free queries $q$ in $\bcq$, the complexity of $\cqa{q}$ is well established by the next theorem.


\begin{theorem}[\cite{KoutrisWTOCS20}] \label{thm:sjf-theorem}
For every self-join-free query~$q$ in $\bcq$, $\cqa{q}$ is 
in \FO, \LSPACE-complete, or \coNP-complete, and it is decidable in polynomial time in the size of $q$ which of the three cases applies. 
\end{theorem}



Past research has indicated that the tools for proving Theorem~\ref{thm:sjf-theorem} largely fall short in dealing with difficulties caused by self-joins.
A notable example concerns \emph{path queries}, i.e., queries of the form:
\begin{equation*}\label{eq:path}
\exists x_{1}\dotsm\exists x_{k+1}\normalformula{R_{1}(\underline{x_{1}},x_{2})\land 
R_{2}(\underline{x_{2}},x_{3})\land\dotsm\land R_{k}(\underline{x_{k}},x_{k+1})}.
\end{equation*}
If a query of this form is self-join-free (i.e., if  $R_{i}\neq R_{j}$ whenever $i\neq j$), then the ``attack graph'' tool~\cite{KoutrisWTOCS20} immediately tells us that $\cqa{q}$ is in~\FO. However, for path queries~$q$ with self-joins, $\cqa{q}$ exhibits a tetrachotomy between \FO, \NL-complete, \PTIME-complete, and \coNP-complete \cite{DBLP:conf/pods/KoutrisOW21}, and the complexity classification requires sophisticated tools.
Note incidentally that self-join-freeness is a simplifying assumption that is also frequent outside CQA (e.g., \cite{FreireGIM15, berkholz2017answering, FreireGIM20,DBLP:journals/tocl/ArenasBM21}).


A natural question is to extend the complexity classification for path queries to  queries that are syntactically less constrained.
In particular, while path queries are restricted to binary relation names, we aim for unrestricted arities, as in practical database systems, which brings us to the construct of tree queries.

A query~$q$ in $\bcq$ is a \emph{rooted (ordered) tree query} if it is uniquely (up to a variable renaming) representable by a rooted ordered tree in which each non-leaf vertex is labeled by a relation name, and each leaf vertex is labeled by a unary relation name, a constant, or~$\bot$. 
The query~$q$ is read from this tree as follows: each vertex labeled by either a relation name or~$\bot$ is first associated with a fresh variable, and each vertex labeled by a constant is associated with that same constant; then, a vertex labeled with relation name $R$ and associated with variable~$x$ represents the query atom $R(\underline{x},y_{1},\ldots,y_{n})$, where $y_{1},\ldots,y_{n}$ are the symbols (variables or constants) associated with the left-to-right ordered children of the vertex~$x$.
The underlined position is the primary key.
Note that a vertex labeled with a relation name of arity $n+1$ must have $n$ children.
For example, consider the rooted tree in Fig.~\ref{fig:intro-query}(a) and associate fresh variables to its vertices as depicted in Fig.~\ref{fig:intro-query}(b). The rooted tree thus represents a query~$q_{1}$ that contains, among others, the atoms $C(\underline{x},y,z)$ and $R(\underline{y},u_{1},v_{1})$. 
It is easy to see that every path query is a rooted tree query. 
The class of all rooted tree queries is denoted $\rtbcq$.
We can now present our main results.

\begin{figure}
\centering
\subfloat[A rooted ordered tree representing~$q_{1}$.]{
\includegraphics{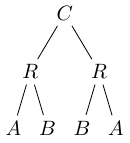}
}  
\qquad 
\subfloat[Each vertex is associated with a fresh variable.]{
\includegraphics{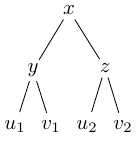}
}
\qquad
\subfloat[A rooted ordered tree representing~$q_{2}$.]{
\includegraphics{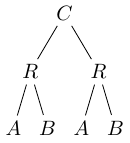}
} 
\caption{The left rooted ordered tree represents (up to a variable renaming) the Boolean conjunctive query $q_{1}$ with atoms  $C(\underline{x}, y, z)$, $R(\underline{y}, u_1, v_1)$, $A(\underline{u_1})$, $B(\underline{v_1})$, $R(\underline{z}, u_2, v_2)$, $B(\underline{u_2})$,  $A(\underline{v_2})$.
The right rooted ordered tree represents $q_2$ with atoms 
$C(\underline{x}, y, z)$, $R(\underline{y}, u_1, v_1)$, $A(\underline{u_1})$, $B(\underline{v_1})$, $R(\underline{z}, u_2, v_2)$, $A(\underline{u_2})$,  $B(\underline{v_2})$.
}
\label{fig:intro-query}
\end{figure}


\begin{theorem}
\label{thm:tree-main}
For every query~$q$ in $\rtbcq$, $\cqa{q}$ is in \FO, \mbox{\NL-hard $\cap$ \LFPL}, or \coNP-complete, and it is decidable in polynomial time in the size of $q$ which of the three cases applies.
\end{theorem}

Here \LFPL\ denotes least fixed point logic as defined in~\cite[p.~181]{DBLP:books/sp/Libkin04} (a.k.a.\ \FO[\LFPL]), and \NL\ denotes the class of problems decidable by a non-deterministic Turing machine using only logarithmic space. 
The classification criteria implied in Theorem~\ref{thm:tree-main} are explicitly stated in Theorem~\ref{thm:trichotomy}.

It will turn out that  subtree homomorphisms play a crucial role in the complexity classification of $\cqa{q}$ for queries~$q$ in $\rtbcq$.
For example, our results show that for the queries $q_{1}$ and~$q_{2}$ represented in, respectively, Fig.~\ref{fig:intro-query}(a) and~(c), $\cqa{q_{1}}$ is \coNP-complete, while $\cqa{q_{2}}$ is in~\FO.
The difference occurs because the two ordered subtrees rooted at $R$ are isomorphic in $q_{2}$ ($A$ precedes $B$ in both subtrees), but not in $q_{1}$.
Another novel and useful tool in the complexity classification is a context-free grammar (CFG) that generalizes the NFA for path queries used in~\cite{DBLP:conf/pods/KoutrisOW21}.

Once Theorem~\ref{thm:tree-main} is proved, it is natural to generalize rooted tree queries further by allowing queries that can be represented by graphs that are not trees.
%

We thereto define $\ibcq$ (Definition~\ref{def:ibcq}), a subclass of $\bcq$ that extends $\rtbcq$.
In $\ibcq$ queries, two distinct atoms can share a variable occurring at non-primary-key positions, which requires representations by DAGs rather than trees. Moreover, $\ibcq$ gives up on the acyclicity requirement that is cooked into $\rtbcq$. 

Significantly, we were able to establish the \FO-boundary in the set $\{\cqa{q}\mid q\in\ibcq\}$. 

\begin{theorem} 
\label{thm:trichotomy:ext:no-berge}
For every query~$q$ in $\ibcq$,
it is decidable whether or not $\cqa{q}$ is in \FO; and when it is, a first-order rewriting can be effectively constructed.
\end{theorem}

So far, we have not achieved a fine-grained complexity classification of all problems in $\{\cqa{q}\mid q\in\ibcq\}$.
However, we were able to do so for the set of Berge-acyclic queries in $\ibcq$, denoted $\aibcq$. Recall that a conjunctive query is Berge-acyclic if its incidence graph (i.e., the undirected bipartite graph that connects every variable~$x$ to all query atoms in which~$x$ occurs) is acyclic.

\begin{theorem} 
\label{thm:trichotomy:ext}
For every query~$q$ in $\aibcq$, 
the problem $\cqa{q}$ is in \FO, \mbox{\NL-hard $\cap$ \LFPL}, or \coNP-complete, and it is decidable in polynomial time in the size of $q$ which of the three cases applies.
\end{theorem}

Since
$\rtbcq\subsetneq\aibcq\subsetneq\ibcq$,
Theorem~\ref{thm:tree-main} is subsumed by Theorem~\ref{thm:trichotomy:ext}. We nevertheless provide Theorem~\ref{thm:tree-main} explicitly, because its proof makes up the main part of this paper.
In Section~\ref{sec:beyond}, we will discuss the challenges in extending Theorem~\ref{thm:trichotomy:ext} beyond $\aibcq$.

\section{Related Work}
\label{sec:related}

Inconsistency management has been studied in various database contexts (e.g., graph databases~\cite{DBLP:journals/jcss/BarceloF17,DBLP:conf/icdt/BarceloF15}, medical databases~\cite{kahale2020meta}, online databases~\cite{katsis2010inconsistency}, spatial databases~\cite{DBLP:journals/is/RodriguezBM13}), and under different repair semantics (e.g., \cite{10.5555/1709465.1709573,DBLP:conf/icdt/LopatenkoB07,DBLP:journals/tods/Wijsen05}).
Arenas, Bertossi, and Chomicki initiated Consistent Query Answering (CQA) in~1999~\cite{10.1145/303976.303983}. 
Twenty years later, their contribution was acknowledged in a \emph{Gems of PODS session}~\cite{Bertossi19}. 
An overview of complexity classification results in CQA appeared in the \emph{Database Principles} column of SIGMOD Record \cite{10.1145/3377391.3377393}.

The term $\cqa{q}$ was coined in~\cite{wijsen2010first} to refer to CQA for Boolean queries $q$ on databases that violate primary keys, one per relation, which are fixed by $q$'s schema. 
The complexity classification of $\cqa{q}$ for the class of self-join-free Boolean conjunctive queries underwent a series of efforts~\cite{FUXMAN2007610,KOLAITIS201277,KoutrisS14,KoutrisW15,DBLP:conf/icdt/KoutrisW19},
until it was revealed that the complexity of $\cqa{q}$ for self-join-free conjunctive queries displays a trichotomy between \FO,  \LSPACE-complete, and \coNP-complete \cite{KoutrisW17,KoutrisWTOCS20}. 
A few extensions beyond this trichotomy result are known. 
Under the requirement of self-join-freeness, it remains decidable whether or not $\cqa{q}$ is in \FO\ in the presence of negated atoms~\cite{KoutrisW18}, multiple keys~\cite{KoutrisW20}, and unary foreign keys~\cite{DBLP:conf/pods/HannulaW22}.

Little is known concerning the complexity classification of the problem $\cqa{q}$ beyond self-join-free conjunctive queries.
For the restricted class of Boolean path queries~$q$, $\cqa{q}$ already exhibits a tetrachotomy between \FO, \NL-complete, \PTIME-complete and \coNP-complete~\cite{DBLP:conf/pods/KoutrisOW21}.
Figueira et~al.~\cite{DBLP:conf/icdt/FigueiraPSS23} have recently discovered a simple fixpoint algorithm that solves $\cqa{q}$ when $q$ is a self-join free conjunctive query or a path query such that $\cqa{q}$  is in \PTIME.
As already discussed in Section~\ref{sec:introduction}, relationships have been found between CQA and CSP~\cite{10.1145/2699912,lutz_et_al:LIPIcs:2015:4995}.

The counting variant of the problem $\cqa{q}$, denoted
$\sharp\cqa{q}$, asks to count the number of repairs that satisfy some Boolean query $q$.
For self-join-free Boolean conjunctive queries, $\sharp\cqa{q}$ exhibits a dichotomy between  \textbf{FP} and \mbox{$\sharp$\PTIME}-complete~\cite{Maslowski2013ADI}. 
This dichotomy has been shown to extend to queries with self-joins if primary keys are singletons~\cite{MaslowskiW14}, and to functional dependencies~\cite{DBLP:conf/pods/CalauttiLPS22a}.
Calautti, Console, and Pieris present in \cite{DBLP:conf/pods/CalauttiCP19} a complexity analysis of these
counting problems under many-one logspace reductions and conducted an experimental evaluation of randomized approximation schemes for approximating the percentage of repairs that satisfy a given query \cite{DBLP:conf/pods/CalauttiCP21}.
CQA is also studied under different notions of repairs like operational repairs \cite{DBLP:conf/pods/CalauttiLP18,DBLP:conf/pods/CalauttiLPS22} and preferred repairs \cite{DBLP:journals/amai/StaworkoCM12,DBLP:journals/tcs/KimelfeldLP20}. 
CQA has also been studied for queries with aggregation, in both theory and practice~\cite{DBLP:conf/sigmod/DixitK21,DBLP:conf/icdt/KhalfiouiW23}.

Theoretical research in CQA has stimulated implementations and
experiments in prototype system, using different target languages and engines: SAT~\cite{DBLP:conf/sat/DixitK19}, ASP~\cite{DBLP:journals/dke/MarileoB10, manna_ricca_terracina_2015,10.1145/3340531.3411911}, BIP~\cite{10.14778/2536336.2536341}, SQL~\cite{DBLP:journals/pacmmod/FanKOW23}, logic programming~\cite{GrecoGZ03}, and hypergraph algorithms~\cite{chomicki2004hippo}.


\section{Preliminaries}
\label{sec:prelim}

We assume disjoint sets of {\em variables\/} and {\em constants\/}.
A {\em valuation\/} over a set $U$ of variables is a total mapping $\theta$ from~$U$ to the set of~constants.


\myparagraph{Atoms and key-equal facts}
Every relation name has a fixed arity, and a fixed set of primary-key positions.
We will underline primary-key positions and assume w.l.o.g.\ that all primary-key positions precede all other positions.
An \emph{atom} is then an expression $R(\underline{s_{1},\ldots,s_{k}},s_{k+1},\ldots,s_{n})$ where each $s_{i}$ is a variable or a constant for $1 \leq i \leq n$.
The sequence $s_{1},\ldots,s_{k}$ is called the \emph{primary key} (of the atom).
This primary key is called \emph{simple} if~$k=1$, and \emph{constant-free} if no constant occurs in it.
An atom without variables is a \emph{fact}.
Two facts are \emph{key-equal} if they use the same relation name and agree on the primary~key.

\myparagraph{Database instances, blocks, and repairs}
A {\em database schema\/} is a finite set of relation names.
All constructs that follow are defined relative to a fixed database schema.
A {\em database instance\/} (or \emph{database} for short) is a finite set $\db$ of facts using only the relation names of the schema.
%
We write $\adom{\db}$ for the active domain of $\db$ (i.e., the set of constants that occur in $\db$).
A {\em block\/} of $\db$ is a maximal set of key-equal facts of $\db$.
Whenever a database instance $\db$ is understood,
we write $R(\underline{\vec{c}},*)$ for the block that contains all facts with relation name~$R$ and primary-key value~$\vec{c}$, where $\vec{c}$ is a sequence of constants.
A database instance $\db$ is {\em consistent\/} if it contains no two distinct facts that are key-equal (i.e., if no block of $\db$ contains more than one fact). 
A {\em repair\/} of $\db$ is an inclusion-maximal consistent subset of $\db$. 
%

\myparagraph{Boolean conjunctive queries}
A {\em Boolean conjunctive query\/} is a finite set 
$q=\{R_{1}(\underline{\vec{x}_{1}},\vec{y}_{1})$, $\dots$, $R_{n}(\underline{\vec{x}_{n}},\vec{y}_{n})\}$ of atoms.
The set $q$ represents the first-order sentence with no free-variables
$$\exists u_{1}\dotsm\exists u_{k}\normalformula{R_{1}(\underline{\vec{x}_{1}},\vec{y}_{1})\land\dotsm\land R_{n}(\underline{\vec{x}_{n}},\vec{y}_{n})},$$ and we denote $\queryvars{q} = \{u_{1}, \dots, u_{k}\}$, the set of variables that occur in~$q$ and denote $\constants{q}$ as the set of constants that occur in $q$. 
We write $\bcq$ for the class of Boolean conjunctive queries.

Let $q$ be a query in $\bcq$.
We say that~$q$ has a {\em self-join\/} if some relation name occurs more than once in $q$.
If~$q$ has no self-joins, it is called {\em self-join-free\/}. 
We say that~$q$ is \emph{minimal} if it is not equivalent to a query in~$\bcq$ with a strictly smaller number of atoms. 


\myparagraph{Consistent query answering}
For every query $q$ in $\bcq$, the decision problem $\cqa{q}$ takes as input a database instance $\db$, and asks whether $q$ is satisfied by every repair of~$\db$.
It is straightforward that $\cqa{q}$ is in \coNP\ for every $q\in\bcq$.



\myparagraph{Rooted relation trees}
A \emph{rooted relation tree} is a (directed) rooted ordered tree where each internal vertex is labeled by a relation name, and each leaf vertex is labeled with either a unary relation name, a constant, or $\bot$, such that every two vertices sharing the same label have the same number of children. 
We denote by $\rootedAt{u}{\tau}$ the subtree rooted at vertex $u$ in~$\tau$. 
Any rooted relation tree~$\tau$ has a string representation recursively defined as follows:

the string representation of a tree with only one vertex is the label of that vertex; 
otherwise, if the root of $\tau$ is labeled $R$ and has the following ordered children $v_1, v_2,\ldots, v_n$, then $\tau$'s string representation is $R(s_1, s_2, \dots, s_n)$, where $s_i$ is the string representation of~$\rootedAt{v_i}{\tau}$. For example, the tree in Fig.~\ref{fig:intro-query}(a) has string representation $C(R(A,B),R(B,A))$.
We will often blur the distinction between rooted relation trees and their string representation.

\myparagraph{Rooted tree query and rooted tree sets}
A \emph{querification} of a rooted relation tree~$\tau$ is a total function~$f$ with domain $\tau$'s vertex set that maps each vertex labeled by a constant to that same constant, and injectively maps all other vertices to variables.
Such a querification naturally extends to a mapping $f(\tau)$ of the entire tree:
if $u$ is a vertex in~$\tau$ with label~$R$ and children $v_1$, $v_2$, \dots, $v_n$, then $f(\tau)$ contains the atom $R(\underline{f(u)}, f(v_1), f(v_2), \dots, f(v_n))$. A Boolean conjunctive query is a \emph{rooted tree query} if it is equal to $f(\tau)$ for some querification~$f$ of some rooted relation tree~$\tau$.
If $q=f(\tau)$, we also say that~$q$ is \emph{represented} by~$\tau$, in which case we often blur the distinction between~$q$ and~$\tau$.
We write $\atomAt{R}{x}{q}$ for the unique $R$-atom in $q$ with primary key variable $x$. 
$\rtbcq$ denotes the class of rooted tree queries.
It can be verified that every rooted tree query is minimal.

Every query $q$ in $\rtbcq$ is represented by a unique rooted relation tree. 
Conversely, every rooted relation tree represents a query in~$\rtbcq$ that is unique up to a variable renaming.
When $f(\tau)=q$, by a slight abuse of terminology, we may use $q$ to refer to~$\tau$, and use the query variable~$x$ (or the expression~$\atomAt{R}{x}{q}$) to refer to the vertex~$u$ in $\tau$ that satisfies $f(u)=x$ and whose label is~$R$.
The variable~$r$ is the \emph{root variable} of a query $q$ in $\rtbcq$ if $r$ is the root vertex of $q$'s rooted relation tree. 
For two distinct vertices $x$ and $y$, we write $x\precedes{q} y$ if the vertex $x$ is an ancestor of $y$ in $q$, and write $x \vardisjoint{q} y$ if neither $x\precedes{q} y$ nor $y \precedes{q} x$. 
When $x$ and $y$ have the same label~$R$, we can also write $\atomAt{R}{x}{q} \precedes{q} \atomAt{R}{y}{q}$ and $\atomAt{R}{x}{q} \diffbranch{q} \atomAt{R}{y}{q}$ instead of $x \precedes{q} y$ and $x\vardisjoint{q} y$ respectively.
For every variable $x$ in a rooted tree query $q$, we write $\rootedAt{x}{q}$ for the subquery of $q$ whose rooted relation tree is the subtree rooted at vertex $x$ in $q$. 
A variable $x$ is a leaf variable in $q$ if $\rootedAt{x}{q} = \bot$, $\rootedAt{x}{q} = c$, or $\rootedAt{x}{q} = A$, for some constant~$c$ or unary relation name~$A$.

An \emph{instantiation} of a rooted relation tree~$\tau$ is a total function~$g$ from $\tau$'s vertex set to constants such that each vertex labeled by a constant~$c$ is mapped to~$c$. 
Such an instantiation naturally extends to a mapping $g(\tau)$ of the entire tree: if $u$ is a vertex in~$\tau$ with label~$R$ and children $v_1$, $v_2$, \dots, $v_n$, then $g(\tau)$ contains the fact $R(\underline{g(u)}, g(v_1), g(v_2), \dots, g(v_n))$.
A subset~$S$ of $\db$ is a \emph{rooted tree set in $\db$ starting in $c$} if $S=g(\tau)$ for some instantiation $g$ of $\tau$ that maps $\tau$'s root to~$c$.
A case of particular interest is when $\db$ is consistent, in particular, when $\db$ is a repair.
It can be verified that a rooted tree set in a repair~$\rep$ is uniquely determined by a constant~$c$ and a rooted tree~$\tau$ (because only one instantiation is possible); by overloading terminology, $\tau$ is also called a rooted tree set in~$\rep$ starting in~$c$. 
For convenience, an empty rooted tree set, denoted by $\bot$, starts in any constant $c$.

\myparagraph{Homomorphism} Let $p,q\in\bcq$. We write $p \leqhomo{}{} q$ if there exists a homomorphism from $p$ to~$q$, i.e., a mapping $h:\vars{p}\rightarrow\vars{q}\cup\constants{q}$ that acts as identity when applied on constants, such that for every atom $R(\underline{\vec{x}}, \vec{y})$ in~$p$, $R(\underline{h(\vec{x})}, h(\vec{y}))$ is an atom of~$q$. 
For $u \in \queryvars{p}$ and $v \in \queryvars{q}$, we write $p \leqhomo{u}{v} q$ if there exists a homomorphism $h$ from $p$ to $q$ with $h(u) = v$.
It can now be verified that for rooted tree queries~$p$ and~$q$, there is a homomorphism~$h$ from $p$ to $q$ if and only if there is a label-preserving graph homomorphism from the rooted relation tree of~$p$ to that of~$q$ (we assume that a leaf vertex with label~$\bot$ can map to a vertex with any label).
Since rooted relation trees are \emph{ordered} trees, graph homomorphisms must evidently be order-preserving.  For example, there is no homomorphism between the trees $R(A,B)$ and $R(B,A)$.

\begin{example}
The following rooted tree query and its rooted relation tree are depicted in Fig.~\ref{fig:rooted-tree-relation-ex}:
\begin{align*}
q = \{&A(\underline{x_0}, x_1, x_2), R(\underline{x_1}, x_3, x_4), R(\underline{x_2}, x_5, x_6), \\
& R(\underline{x_3}, x_7, x_8), U(\underline{x_7}), \\
& X(\underline{x_4}, c_1), Y(\underline{x_5}, x_{9}), Z(\underline{x_6}, c_2, x_{10})\}.
\end{align*}
We have:
\begin{align*}
\rootedAt{x_1}{q} &= R(\underline{x_1}, x_3, x_4), R(\underline{x_3}, x_7, x_8), U(\underline{x_7}), X(\underline{x_{4}},c_{1})\\
&= R(R(U,\bot),X(c_1)),\\
\rootedAt{x_2}{q} &= R(\underline{x_2}, x_5, x_6), Y(\underline{x_5}, x_{9}), Z(\underline{x_6}, c_2, x_{10})\\
&= R(Y(\bot),Z(c_2,\bot)),\\
\rootedAt{x_3}{q} &= R(\underline{x_3}, x_7, x_8), U(\underline{x_7}) \\
&= R(U,\bot).
\end{align*}
In this query $q$, we have $\atomAt{R}{x_1}{q} \vardisjoint{q} \atomAt{R}{x_2}{q}$, $\atomAt{R}{x_1}{q} \precedes{q} \atomAt{R}{x_3}{q}$, and $\atomAt{R}{x_2}{q} \vardisjoint{q} \atomAt{R}{x_3}{q}$.

\begin{figure}
\centering
\subfloat[A rooted relation tree $\tau$]{
\includegraphics{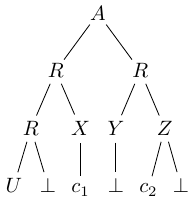}
}  
\quad
\subfloat[A mapping $f$ from vertices in $\tau$ to variables in $q$]{
\includegraphics{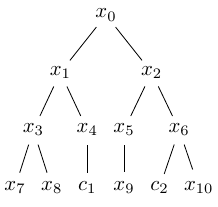}
}
\caption{An example rooted relation tree, where $c_1$ and $c_2$ are constants.}
\label{fig:rooted-tree-relation-ex}
\end{figure}
\end{example}


\section{The Complexity Classification}
\label{sec:trichotomy}


Our classification focuses on rooted tree queries ($\rtbcq$). 
We will extend to $\aibcq$ and $\ibcq$ in Section~\ref{sec:berge}.
The classification of path queries in~\cite{DBLP:conf/pods/KoutrisOW21} uses a notion of ``rewinding'' to deal with self-joins: a path query $u\cdot Rv \cdot Rw$ rewinds to $u\cdot Rv \cdot Rv \cdot Rw$. 
Very informally, rewinding captures that query atoms with the same relation name can be ``confused'' with one another (or ``rewind'' to one another in our terminology) during query evaluation: in $u\cdot Rv\cdot Rw$, once we have evaluated the prefix $u\cdot Rv\cdot R$, the last $R$ can be confused with the first one, in which case we continue with the suffix $Rv \cdot Rw$ (instead of merely $Rw$).
We generalize the notion of rewinding from path queries to rooted tree queries.

\begin{definition}[Rewinding]
\label{definition:rewinding}
Let $q$ be a query in $\rtbcq$. 
Let $R(\underline{x}, \dots)$ and $R(\underline{y}, \dots)$ be two (not necessarily distinct) atoms in $q$.
We define $\rewind{q}{y}{x}{R}$ as the following rooted tree query
$$\rewind{q}{y}{x}{R} := \formula{q \setminus \rootedAt{y}{q}} \cup f(\rootedAt{x}{q}),$$
for some isomorphism $f$ that maps $x$ to $y$ (i.e., $f(x)=y$), and maps every other variable in $\rootedAt{x}{q}$ to a fresh variable.
\end{definition}

Intuitively, the rooted tree query $\rewind{q}{y}{x}{R}$ can be obtained by replacing $\rootedAt{y}{q}$ with a fresh copy of $\rootedAt{x}{q}$.
Fig.~\ref{fig:rewinding-example} presents some rooted tree queries obtained from rewinding on the rooted tree $q$ in Fig.~\ref{fig:rooted-tree-relation-ex}.
\begin{figure*}
\centering
\subfloat[$\rewind{q}{x_1}{x_2}{R}$]{
\includegraphics{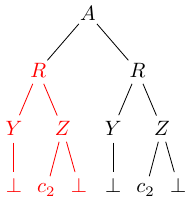}
}
\quad
\subfloat[$\rewind{q}{x_2}{x_1}{R}$]{
\includegraphics{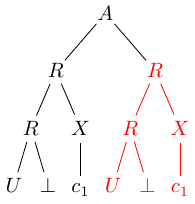} 
}
\quad
\subfloat[$\rewind{q}{x_3}{x_1}{R}$]{
\includegraphics{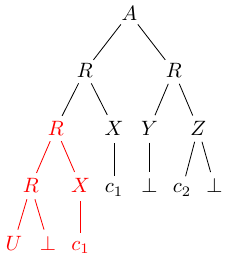} 
}
\caption{An illustration of rewinding for the query of Fig.~\ref{fig:rooted-tree-relation-ex}; the modified subtrees are highlighted in red.}
\label{fig:rewinding-example}
\end{figure*}

The classification criteria in \cite{DBLP:conf/pods/KoutrisOW21} uses the notions of factors and prefixes that are specific to words, which can be generalized using homomorphism on rooted tree queries.
Consider the following syntactic conditions on a rooted tree query $q$ with root variable $r$: 

\begin{itemize}
\item $\ctwo$ : for every two atoms  $R(\underline{x}, \dots)$ and $R(\underline{y}, \dots)$ in $q$, either $q\leqhomo{}{}\rewind{q}{y}{x}{R}$ or $q\leqhomo{}{}\rewind{q}{x}{y}{R}$.
\item $\cone$ : for every two atoms  $R(\underline{x}, \dots)$ and $R(\underline{y}, \dots)$ in $q$, either $q\leqhomo{r}{r}\rewind{q}{y}{x}{R}$ or $q\leqhomo{r}{r}\rewind{q}{x}{y}{R}$.
\end{itemize}

It is easy to see that conditions $\cone$ and $\ctwo$ are decidable in polynomial time in the size of the query.
We may restate $\ctwo$ and $\cone$ using more fine-grained syntactic conditions below.
\begin{itemize}
\item $\cbranch$ : for every two atoms $\atomAt{R}{x}{q} \vardisjoint{q} \atomAt{R}{y}{q}$ in $q$,
either $\rootedAt{y}{q} \leqhomo{y}{x} \rootedAt{x}{q}$ or 
$\rootedAt{x}{q} \leqhomo{x}{y} \rootedAt{y}{q}$.
\item $\cfactor$ : for every two atoms $\atomAt{R}{x}{q} \precedes{q} \atomAt{R}{y}{q}$ in $q$, we have $q\leqhomo{}{}\rewind{q}{y}{x}{R}$.
\item $\cprefix$ : for every two atoms $\atomAt{R}{x}{q} \precedes{q} \atomAt{R}{y}{q}$ in $q$, we have $q \leqhomo{r}{r} \rewind{q}{y}{x}{R}$.
\end{itemize}

\begin{lemma}
\label{lemma:branch-to-root-homomorphism}
For every two atoms $\atomAt{R}{x}{q} \vardisjoint{q} \atomAt{R}{y}{q}$ in a rooted tree query $q$, we have
$q \leqhomo{}{} \rewind{q}{y}{x}{R}$ if and only if $\rootedAt{y}{q} \leqhomo{y}{x} \rootedAt{x}{q}$.
\end{lemma}

For the sake of simplicity, we postpone the proof of Lemma~\ref{lemma:branch-to-root-homomorphism} to Appendix~\ref{app:trichotomy}. 
Lemma~\ref{lemma:branch-to-root-homomorphism} implies the following connections among the syntactic conditions.


\begin{proposition}
$\ctwo = \cfactor\land\cbranch$, $\cone=\cprefix\land\cbranch$.
\end{proposition}

\begin{figure*}
\centering
\subfloat[$\cone$]{
\includegraphics{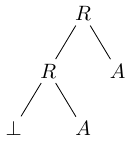}
}
\quad
\subfloat[$\ctwo, \lnot \cprefix$]{
\includegraphics{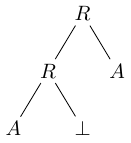}  
}
\quad
\subfloat[$\lnot \cbranch$]{
\includegraphics{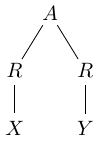}  
}
\quad
\subfloat[$\cone$]{
\includegraphics{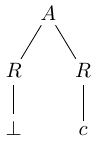}    
}
\quad
\subfloat[$\cone$]{
\includegraphics{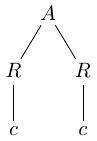}
   
}
\quad
\subfloat[$\lnot \cbranch$]{
\includegraphics{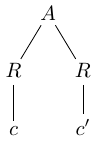}   
}
\quad
\subfloat[$\lnot \cfactor$]{
\includegraphics{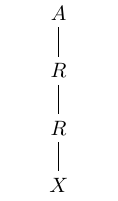}  
}
\caption{Examples of rooted relation trees. Trees annotated with $\lnot\mathsf{C}$ violate syntactic condition $\mathsf{C}$, while trees annotated with~$\mathsf{C}$ satisfy $\mathsf{C}$. For example, the tree in~(a) satisfies $\cone$; and the tree~(b) satisfies $\ctwo$ but violates $\cprefix$.}
\label{fig:tree-examples}
\end{figure*}


\begin{example}
\label{ex:tree-examples}
Let $q$ be as in Fig.~\ref{fig:rooted-tree-relation-ex}. We have that $q$ violates $\cbranch$ (and therefore $\ctwo$), since there is no homomorphism from $q$ to neither $\rewind{q}{x_1}{x_2}{R}$ nor $\rewind{q}{x_2}{x_1}{R}$.

Fig.~\ref{fig:tree-examples} shows some example rooted relation trees annotated with the syntactic conditions they satisfy or violate.
\end{example}

Our main classification result can now be stated.

\begin{theorem} [Trichotomy Theorem]
\label{thm:trichotomy}
For every query~$q$ in $\rtbcq$, 
\begin{itemize}
    \item if $q$ satisfies $\ctwo$, then the problem $\cqa{q}$ is in \LFPL; otherwise it is \coNP-complete; and
    \item if $q$ satisfies $\cone$, then the problem $\cqa{q}$ is in \FO;  otherwise it is \NL-hard.
\end{itemize}
\end{theorem}

Let us provide some intuitions behind Theorem~\ref{thm:trichotomy}.
Both $\cprefix$ and $\cfactor$ concern the homomorphism from $q$ to the rooted tree query obtained by rewinding from a subtree to its ancestor subtree, which resembles the case on path queries. 
The condition $\cbranch$ is vacuously satisfied for path queries, but is crucial to the classification of rooted tree queries.

For the complexity lower bound, if $q$ violates $\cbranch$, then $\cqa{q}$ is \coNP-hard. 
Intuitively, this is because if $\rootedAt{x}{q}$ and $\rootedAt{y}{q}$ are not homomorphically comparable and appear in different branches, then the facts in their common ancestor relation may ``choose'' which branch to satisfy, which allows us to reduce from \textsf{SAT} in item~\eqref{it:conphardness} of Proposition~\ref{prop:hardness}.
For example, consider the query $q_1$ as in Fig.~\ref{fig:intro-query}(a) and the example database instance~$\db$ in Fig.~\ref{tbl:instance}.
It can be shown that there is a repair of $\db$ that falsifies~$q_1$ if and only if the following CNF formula is satisfiable:
$$\underbrace{(x_1 \lor x_2)}_{C_1} \land \underbrace{(\overline{x_1} \lor \overline{x_2})}_{C_2}.$$

For the complexity upper bound, if $\rootedAt{y}{q} \leqhomo{y}{x} \rootedAt{x}{q}$, the arguments above fail because the facts in their common ancestor relation cannot ``choose'' which branch to satisfy anymore: informally, whenever $\rootedAt{x}{q}$ is satisfied, $\rootedAt{y}{q}$ will be satisfied due to the homomorphism.
This crucial observation from $\cbranch$ also leads to a total preorder on all self-joining atoms, which allows us to deal with self-joining atoms in different branches as if they were on a path.

\begin{figure}
    \centering
    $$
\begin{array}{ccc}
\begin{array}{c|ccc}
C & \underline{1} & 2 & 3 \bigstrut\\
\cline{2-4}
\ast  & c_1 & x_1 & z_-\\
  & c_1 & x_2 & z_-\\
\cdashline{2-4}
  & c_2 & z_+ & x_1 \\
\ast  & c_2 & z_+ & x_2
\end{array}
&
\begin{array}{c|ccc}
R & \underline{1} & 2 & 3\bigstrut\\
\cline{2-4}
  & x_1 & a & b\\
\ast  & x_1 & b & a\\
\cdashline{2-4}
\ast  & x_2 & a & b\\
  & x_2 & b & a\\
\cdashline{2-4}
\ast  & z_+ & a & b\\
\cdashline{2-4}
\ast  & z_- & b & a
\end{array}
&
\begin{array}{c|c}
A & \underline{1}\\
\cline{2-2}
\ast  & a \\
\multicolumn{2}{}{} \\
B & \underline{1}\\
\cline{2-2}
\ast  & b
\end{array}
\end{array}
$$
\caption{An inconsistent database instance $\db$ for $\cqa{q_1}$, where $q_1$ is represented in Fig.~\ref{fig:intro-query}(a).
Blocks are separated by dashed lines.
The facts with~$\ast$ form a repair that falsifies~$q_{1}$, corresponding to a satisfying truth assignment $x_{1}=1$ and $x_{2}=0$.
}
\label{tbl:instance}
\end{figure}

\begin{definition}[Relation $\mytotalorder{q}$]
\label{def:totalorder}
Let $q$ be a query in $\rtbcq$. 
Let $\atomAt{R}{x}{q}$ and $\atomAt{R}{y}{q}$ be two atoms in $q$. We write $\atomAt{R}{x}{q} \mytotalorder{q} \atomAt{R}{y}{q}$ if either $\atomAt{R}{x}{q} \precedes{q} \atomAt{R}{y}{q}$ or $\rootedAt{y}{q} \leqhomo{y}{x} \rootedAt{x}{q}$.
\end{definition}

\begin{proposition}
\label{prop:totalorder}
Let $q$ be a query in $\rtbcq$ satisfying $\cbranch$. For every relation name $R$, the relation $\mytotalorder{q}$ is a total preorder on all $R$-atoms in $q$.
\end{proposition}
\begin{proof}
We first show that every two distinct atoms $\atomAt{R}{x}{q}$ and $\atomAt{R}{x}{q}$ are comparable by $\mytotalorder{q}$.
Let $\atomAt{R}{x}{q}$ and $\atomAt{R}{y}{q}$ be two distinct atoms in $q$. 
The claim holds if $\atomAt{R}{x}{q} \precedes{q} \atomAt{R}{y}{q}$ or $\atomAt{R}{y}{q} \precedes{q} \atomAt{R}{x}{q}$. 
Otherwise, we have $\atomAt{R}{x}{q} \vardisjoint{q} \atomAt{R}{y}{q}$, and since $q$ satisfies $\cbranch$, we have either 
$\rootedAt{x}{q} \leqhomo{x}{y} \rootedAt{y}{q}$ or $\rootedAt{y}{q} \leqhomo{y}{x} \rootedAt{x}{q}$, as desired.

Next we show show that $\mytotalorder{q}$ is transitive. 
Assume $\atomAt{R}{x}{q} \mytotalorder{q} \atomAt{R}{y}{q}$ and $\atomAt{R}{y}{q} \mytotalorder{q} \atomAt{R}{z}{q}$. We distinguish four cases.

\begin{itemize}
\item Case that $\atomAt{R}{x}{q} \precedes{q} \atomAt{R}{y}{q}$ and $\atomAt{R}{y}{q} \precedes{q} \atomAt{R}{z}{q}$. 
Then we have $\atomAt{R}{x}{q} \precedes{q} \atomAt{R}{z}{q}$, as desired.

\item Case that $\rootedAt{y}{q} \leqhomo{y}{x} \rootedAt{x}{q}$ and $\rootedAt{z}{q} \leqhomo{z}{y} \rootedAt{y}{q}$. 
Then we have $\rootedAt{z}{q} \leqhomo{z}{x} \rootedAt{x}{q}$, as desired.

\item Case that $\atomAt{R}{x}{q} \precedes{q} \atomAt{R}{y}{q}$ and $\rootedAt{z}{q} \leqhomo{z}{y} \rootedAt{y}{q}$. 
The claim follows if $\atomAt{R}{x}{q} \precedes{q} \atomAt{R}{z}{q}$. 
Suppose for contradiction that $\atomAt{R}{z}{q} \precedes{q} \atomAt{R}{x}{q}$. 
Then $\atomAt{R}{z}{q} \precedes{q} \atomAt{R}{y}{q}$, and $\rootedAt{z}{q}$ contains more atoms than $\rootedAt{y}{q}$. However, we have $\rootedAt{z}{q} \leqhomo{z}{y} \rootedAt{y}{q}$, a contradiction. 
It then must be that $\atomAt{R}{x}{q} \vardisjoint{q} \atomAt{R}{z}{q}$. 
Suppose for contradiction that $\rootedAt{x}{q} \leqhomo{x}{z} \rootedAt{z}{q}$. 
Then we have $\rootedAt{x}{q} \leqhomo{x}{y} \rootedAt{y}{q}$, but $\atomAt{R}{x}{q} \precedes{q} \atomAt{R}{y}{q}$, a contradiction. 
Since $q$ satisfies $\cbranch$, we have $\rootedAt{z}{q} \leqhomo{z}{x} \rootedAt{x}{q}$, as desired.

\item Case that $\rootedAt{y}{q} \leqhomo{y}{x} \rootedAt{x}{q}$ and $\atomAt{R}{y}{q} \precedes{q} \atomAt{R}{z}{q}$. 
The claim follows if $\atomAt{R}{x}{q} \precedes{q} \atomAt{R}{z}{q}$. 
Suppose for contradiction that $\atomAt{R}{z}{q} \precedes{q} \atomAt{R}{x}{q}$. 
Then $\atomAt{R}{y}{q} \precedes{q} \atomAt{R}{x}{q}$, and $\rootedAt{y}{q}$ contains more atoms than $\rootedAt{x}{q}$. However, we have $\rootedAt{y}{q} \leqhomo{y}{x} \rootedAt{x}{q}$, a contradiction. 
It then must be that $\atomAt{R}{x}{q} \vardisjoint{q} \atomAt{R}{z}{q}$. 
Suppose for contradiction that $\rootedAt{x}{q} \leqhomo{x}{z} \rootedAt{z}{q}$. 
Then we have $\rootedAt{y}{q} \leqhomo{y}{z} \rootedAt{z}{q}$, but $\atomAt{R}{y}{q} \precedes{q} \atomAt{R}{z}{q}$, a contradiction. 
Since $q$ satisfies $\cbranch$, it follows that $\rootedAt{z}{q} \leqhomo{z}{x} \rootedAt{x}{q}$.
\end{itemize}
This concludes the proof.
\end{proof}


The remainder of this paper is organized as follows. 
Section~\ref{sec:treenfa} defines a context-free grammar $\treecfg{q}$ for each $q\in\rtbcq$, and the problem $\certaintrace{q}$ that concerns $\treecfg{q}$.  Lemma~\ref{lem:min-start} concludes the equivalence of $\cqa{q}$ and $\certaintrace{q}$ if~$q$ satisfies $\ctwo$ (or $\cone$).
Section~\ref{sec:certaintrace-ptime}, we show that $\certaintrace{q}$ is in \LFPL~(and in \PTIME) if $q$ satisfies $\cbranch$.
In Sections~\ref{sec:algorithm} and~\ref{sec:hardness}, we show the upper bounds and lower bounds in Theorem~\ref{thm:trichotomy} respectively.
In Section~\ref{sec:berge}, we prove Theorems~\ref{thm:trichotomy:ext:no-berge} and~\ref{thm:trichotomy:ext}.

\section{Context-Free Grammar}
\label{sec:treenfa}

We first generalize NFAs used in the study of path queries~\cite{DBLP:conf/pods/KoutrisOW21} to context-free grammars (CFGs).

\begin{definition}[$\treecfg{q}$]
\label{defn:cfg}
Let $q$ be a query in $\rtbcq$ with root variable $r$.
We define a context-free grammar $\treecfg{q}$ over the string representations of rooted relation trees for each rooted tree query $q$.
The alphabet $\Sigma$ of $\treecfg{q}$ contains every relation symbol and constant in~$q$, open/close parentheses, $\bot$ and comma.

Whenever $v$ is a variable or a constant in~$q$, there is a nonterminal symbol $S_v$.
Every symbol in $\Sigma$ is a terminal symbol. The rules of $\treecfg{q}$ are as follows:
\begin{itemize}
\item for each atom $\atomAt{R}{y}{q} = R(\underline{y}, y_1, y_2, \dots, y_n)$ in $q$, there is a forward production rule 
\begin{equation} \label{eq:forward}
S_y \rightarrow_q R(S_{y_1}, S_{y_2}, \dots, S_{y_n})
\end{equation}
\item whenever $\atomAt{R}{x}{q}$ and $\atomAt{R}{y}{q}$ are atoms in~$q$ such that $\atomAt{R}{x}{q} \precedes{q} \atomAt{R}{y}{q}$, there is a backward production rule 

\begin{equation} \label{eq:backward}
S_y \rightarrow_q S_x
\end{equation}
\item for every leaf variable $u$ whose label~$L$ is either $\bot$ or a unary relation name, there is a rule 
\begin{equation} 
S_u \rightarrow_q L
\end{equation}
\item for each constant $c$ in $q$, there is a rule 
\begin{equation}
S_c \rightarrow_q c
\end{equation}
\end{itemize}

The starting symbol of $\treecfg{q}$ is $S_r$ where $r$ is the root variable of $q$. 
A rooted relation tree $\tau$ is accepted by $\treecfg{q}$, denoted as $\tau \in \treecfg{q}$, if the string representation of $\tau$ can be derived from $S_r$, written as $S_r \derive_q \tau$.
\end{definition}

\begin{example}
Let $q$ be as in Fig.~\ref{fig:rooted-tree-relation-ex}(a) with variables labeled as in Fig.~\ref{fig:rooted-tree-relation-ex}(b). 
The rooted relation tree $\tau$ in Fig.~\ref{fig:rewinding-example}(c) has string representation
$\tau=A(\tau_1,\tau_2)$ where 
\begin{align*}
\tau_1 &= R(R(R(U,\bot),X(c_1)),X(c_1)), \\
\tau_2 &= R(Y(\bot),Z(c_2,\bot)).
\end{align*}
We have $S_{x_2} \derive_q \tau_2$ by applying only forward rewrite rules.
We show next $S_{x_1} \derive_q \tau_1$, using a backward rewrite rule $S_{x_{3}} \rightarrow_q S_{x_{1}}$ at some point:

\begin{align*}
S_{x_1}
&\rightarrow_q R(S_{x_3}, S_{x_4}) \\
&\rightarrow_q R(S_{x_1}, X(S_{c_1})) \\
&\rightarrow_q R(R(S_{x_3}, S_{x_4}), X(c_1)) \\
&\rightarrow_q R(R(R (S_{x_7}, S_{x_8}), X(S_{c_1})), X(c_1)) \\
&\rightarrow_q R(    R(R(U,\bot)),X(c_1)),      X(c_1)    ) = \tau_1.
\end{align*}

Thus
$
S_{x_0} \rightarrow_q A(S_{x_1}, S_{x_2}) \derive_q A(\tau_1,\tau_2) = \tau.
$
So it is correct to conclude that $\tau$ is accepted by $\treecfg{q}$.
\end{example}


Recall from Section~\ref{sec:prelim} that a rooted tree set in a repair~$\rep$ is uniquely determined by a rooted tree~$\tau$ and a constant~$c$; such a rooted tree set is said to be accepted by~$\treecfg{q}$ if $\tau\in\treecfg{q}$.
For our technical treatment later, we next define modifications of $\treecfg{q}$ by changing its starting terminal.


\begin{definition}[$\streecfg{q}{u}$] \label{definition:partial-treelanguage}
For a query $q$ in $\rtbcq$ and a variable $u$ in $q$,
we define $\streecfg{q}{u}$ as the context-free grammar that accepts a rooted relation tree $\tau$ if and only if $S_u \derive_q \tau$. 
\end{definition}



We now introduce the \emph{certain trace problem}.
For each $q$ in $\rtbcq$, $\certaintrace{q}$ is defined as the following decision problem:

\begin{itemize}
\item[] \textbf{PROBLEM} $\certaintrace{q}$ 
\item[] \textbf{Input}: A database instance $\db$.
\item[] \textbf{Question}: Is there a constant $c \in \db$, such that for every repair $\rep$ of $\db$, there is a rooted tree set $\tau$ in $\rep$ starting in $c$ with $\tau\in\treecfg{q}$?
\end{itemize}

The problems $\cqa{q}$ and $\certaintrace{q}$ reduce to each other if $q$ satisfies $\ctwo$. 

\begin{lemma}
\label{lem:min-start}
Let $q$ be a query in $\rtbcq$ satisfying $\ctwo$. Let $\db$ be a database instance. Then the following statements are equivalent:
\begin{enumerate}
\item $\db$ is a ``yes''-instance of $\cqa{q}$; and 
\label{it:rifidb}
\item $\db$ is a ``yes''-instance of $\certaintrace{q}$. 
\label{it:rificonstant}
\end{enumerate}
\end{lemma}

The proof of Lemma~\ref{lem:min-start} is deferred to Section~\ref{sec:algorithm} since it requires some useful results to be developed in the subsequent sections.

\section{Membership of $\certaintrace{q}$ in \LFPL}
\label{sec:certaintrace-ptime}

In this section, we show that the problem $\certaintrace{q}$ is expressible in \LFPL~(and thus in \PTIME)~if $q$ satisfies $\cbranch$.
Let $\db$ be a database instance. 
Consider the algorithm in Fig.~\ref{fig:algo}, following a dynamic programming fashion. 
The algorithm iteratively computes a set $B$ of pairs $\pair{c}{y}$ until it reaches a fixpoint, ensuring that
\begin{quote}
whenever $\pair{c}{y}$ is added to $B$, then every repair of~$\db$ contains a rooted tree set starting in $c$ that is accepted by $\streecfg{q}{y}$. 
\end{quote}
Intuitively, this holds true because $\pair{c}{y}$ is added to $B$ if for every possible fact $f=R(\underline{c}, \vec{d})$ that can be chosen by a repair of $\db$, the context-free grammar $\streecfg{q}{y}$ can proceed 
by firing forward rule with nonterminal $S_y$ that consumes $f$ from the rooted tree set, or by non-deterministically firing some backward rule of the form $S_y\rightarrow_{q}S_{x}$.

The formal semantics for each pair $\pair{c}{y}$ is stated in Lemma~\ref{lemma:correctness}.

\begin{figure*}[t]
\centering
\begin{tabular}{rl}
\textbf{Initialization Step:} & 
\textbf{for every }$c \in \adom{\db}$ and leaf variable or constant $u$ in $q$ \\
&
\begin{tabular}[t]{rll}
& \textbf{add} $\pair{c}{u}$ to $B$ \textbf{if} & $u=c$ is a constant, \\
& & or the label of variable $u$ in $q$ is either $\bot$, \\ 
& & or $L$ with $L(\underline{c}) \in \db$.
\end{tabular} \\
\textbf{Iterative Rule:} & 
\textbf{for every }$c \in \adom{\db}$ and atom $R(\underline{y}, y_1,y_2,\dots,y_n)$ in $q$ \\
& 
\begin{tabular}[t]{rl}
   & \textbf{add} $\pair{c}{y}$ to $B$ \textbf{if} the following formula holds:
\end{tabular}  
\end{tabular}
\begin{minipage}{\textwidth}   
$$
\exists \vec{d} : R(\underline{c}, \vec{d}) \in \db \land \forall \vec{d} : \formula{
R(\underline{c}, \vec{d}) \in \db \rightarrow \factformula{R}{c}{\vec{d}}{y}},
$$where 
$$
\begin{array}{l}
\factformula{R}{c}{\vec{d}}{y} = 
\underbrace{\formula{\displaystyle{\bigwedge_{1 \leq i \leq n}} \pair{d_i}{y_i} \in B}}_{\text{forward production}}
\lor 
\underbrace{\formula{\displaystyle{\bigvee_{\atomAt{R}{x}{q} \precedes{q} \atomAt{R}{y}{q}}}  \factformula{R}{c}{\vec{d}}{x}}}_{\text{backward production}}
\end{array}
$$
and $\vec{d} = \langle d_1, d_2, \dots, d_n \rangle$.
\end{minipage}
\caption{A fixpoint algorithm for computing a set $B$, for a fixed rooted tree~$q$. 
}
\label{fig:algo}
\end{figure*}

\begin{lemma}
\label{lemma:correctness}
Let $q$ be a query in $\rtbcq$ satisfying $\cbranch$.
Let $\db$ be a database instance. 
Let $B$ be the output of the algorithm in Fig.~\ref{fig:algo}. 
Then for any constant $c\in\adom{\db}$ and a variable $y$ in $q$, the following statements are equivalent:
\begin{enumerate}
\item $\pair{c}{y} \in B$; and
\label{it:syntax}
\item for every repair $\rep$ of $\db$, there exists a rooted tree set $\tau$ in~$\rep$ starting in~$c$ such that $\tau\in\streecfg{q}{y}$.
\label{it:semantics}
\end{enumerate}
\end{lemma}

The crux in the proof of Lemma~\ref{lemma:correctness} relies on the existence of repairs called \emph{frugal}: to show item~(\ref{it:semantics}) of Lemma~\ref{lemma:correctness}, it will be sufficient to show that it holds true for frugal repairs.
Frugal repairs also turn out to be useful in proving Lemma~\ref{lem:min-start} and offer an alternative perspective to the algorithm, as stated in Corollary~\ref{prop:frugal-perspective}.

\subsection{Frugal repairs}
\label{sec:frugal-repairs}

We first show that the evaluation result of the predicate ``$\factformulapred$'' and the membership in $B$ in the algorithm of Fig.~\ref{fig:algo} propagate along the total preorder $\mytotalorder{q}$.

\begin{lemma}
\label{lemma:chain-property}
Let $q$ be a query in $\rtbcq$ satisfying $\cbranch$, and~$\db$ a database instance. 
Let $\atomAt{R}{x}{q}, \atomAt{R}{y}{q}$ be two atoms of~$q$.
Then for every fact $R(\underline{c},\vec{d})$ in $\db$ and two atoms $\atomAt{R}{x}{q} \mytotalorder{q} \atomAt{R}{y}{q}$,
\begin{enumerate}
\item
if $\factformula{R}{c}{\vec{d}}{x}$ is true, then $\factformula{R}{c}{\vec{d}}{y}$ is true, where $\factformulapred$ is defined by the algorithm of Fig.~\ref{fig:algo}; and
\item
if $\pair{c}{x} \in B$, then $\pair{c}{y} \in B$, where $B$ is the output of the algorithm of Fig.~\ref{fig:algo}; and
\end{enumerate}
\end{lemma}

The technical proof of Lemma~\ref{lemma:chain-property} is deferred to Appendix~\ref{app:certaintrace-ptime}.

\begin{definition}[Frugal Set]
Let $q$ be a query in $\rtbcq$ satisfying $\cbranch$, and~$\db$ a database instance. 
Let $f = R(\underline{c}, \vec{d})$ be an $R$-fact in $\db$. 
We define the frugal set of $f$ in $\db$ with respect to $q$ as $$\sset{f}{\db}{q} = \{\atomAt{R}{x}{q} \in \atoms{q} \mid \factformula{R}{c}{\vec{d}}{x} \text{ is true}\}.$$
\end{definition}

\begin{lemma}
\label{lemma:comparable}
Let $q$ be a query in $\rtbcq$ satisfying $\cbranch$, and~$\db$ a database instance. 
For every two key-equal facts $f$ and~$g$ in~$\db$, the sets $\sset{f}{\db}{q}$ and $\sset{g}{\db}{q}$ are comparable by~$\subseteq$.
\end{lemma}
\begin{proof}
Suppose for contradiction that there exists two key-equal facts $f=R(\underline{c}, \vec{d_1})$ and $g=R(\underline{c}, \vec{d_2})$ in $\db$ such that $\atomAt{R}{x}{q} \in \sset{f}{\db}{q} \setminus \sset{g}{\db}{q}$ and $\atomAt{R}{y}{q} \in \sset{g}{\db}{q} \setminus \sset{f}{\db}{q}$.
By Proposition~\ref{prop:totalorder}, assume without loss of generality that $\atomAt{R}{x}{q}\mytotalorder{q}\atomAt{R}{y}{q}$.
Then since $\atomAt{R}{x}{q} \in \sset{f}{\db}{q}$, we have $\factformula{R}{c}{\vec{d_1}}{x}$ is true, and thus $\factformula{R}{c}{\vec{d_1}}{y}$ is true by Lemma~\ref{lemma:chain-property}, and hence $\atomAt{R}{y}{q} \in \sset{f}{\db}{q}$, a contradiction.
A similar contradiction can also be reached if $\atomAt{R}{y}{q}\mytotalorder{q}\atomAt{R}{x}{q}$.
This completes the proof.
\end{proof}

Informally, by Lemma~\ref{lemma:comparable}, 
among all facts of a non-empty block $R(\underline{c},*)$ in $\db$, 
there is a (not necessarily unique) fact $R(\underline{c},\vec{d})$ with a $\subseteq$-minimal frugal set in $\db$.
The repair~$\rep^*$ of $\db$ containing all such facts is frugal in the sense that each fact in it satisfies as few $R$-atoms as possible; 
and if $\rep^*$ contains a rooted tree set $\tau$ starting in~$c$ accepted by $\streecfg{q}{y}$, so should every repair of~$\db$. We now formalize this idea.


\begin{definition}[Frugal repair]
\label{def:frugal-repair}
Let $q$ be a query in $\rtbcq$ satisfying $\cbranch$.
Let $\db$ be a database instance.
A \emph{frugal repair}~$\rep^*$ of $\db$ with respect to $q$ is constructed by picking, from each block $R(\underline{c}, *)$ of $\db$,  a fact $R(\underline{c}, \vec{d})$  which $\subseteq$-minimizes $\sset{R(\underline{c}, \vec{d})}{\db}{q}$.
\end{definition}

Lemma~\ref{lemma:frugal-fact-choice-static} is then straightforward by construction of a frugal repair.

\begin{lemma}
\label{lemma:frugal-fact-choice-static}
Let $q$ be a rooted tree query satisfying $\cbranch$.
Let $\db$ be a database instance.
Let $\rep^*$ be a frugal repair of $\db$ with respect to $q$ and let $R(\underline{c}, \vec{d}) \in \rep^*$.
Let $\atomAt{R}{u}{q}$ be an atom in $q$.
If $\factformula{R}{c}{\vec{d}}{u}$ is true, then $\pair{c}{u}\in B$.
\end{lemma}
\begin{proof}
Let $R(\underline{c}, \vec{b})$ be an arbitrary fact in the block $R(\underline{c}, *)$ in $\db$.
By construction of a frugal repair, we have that $\sset{R(\underline{c},\vec{d})}{\db}{q} \subseteq \sset{R(\underline{c}, \vec{b})}{\db}{q}$.
Since $R(\underline{c}, \vec{d}) \in \rep^*$ and $\factformula{R}{c}{\vec{d}}{u}$ is true, we have $\atomAt{R}{u}{q} \in \sset{R(\underline{c},\vec{d})}{\db}{q}$. 
Thus, $\atomAt{R}{u}{q}\in\sset{R(\underline{c},\vec{b})}{\db}{q}$ and 
$\factformula{R}{c}{\vec{b}}{u}$ is true.
Hence $\pair{c}{u} \in B$.
\end{proof}

Lemma~\ref{lemma:frugal-no-tree} shows a desirable property of frugal repairs. 

\begin{lemma}
\label{lemma:frugal-no-tree}
Let $q$ be a query in $\rtbcq$ satisfying $\cbranch$. Let $\db$ be a database instance.
Let $\rep^*$ be a frugal repair of $\db$ with respect to $q$.
If there is a rooted tree set $\tau$ in $\rep^*$ starting in $c$ such that $\tau \in \streecfg{q}{y}$, then $\pair{c}{y} \in B$.
\end{lemma}
\begin{proof}
Let $\tau$ be a rooted tree set starting in $c$ in $\rep^*$ such that $\tau \in \streecfg{q}{y}$. We recursively define a tree trace $\mathcal{T}$ on nodes of the form $(c, x)$, where $c\in\adom{\rep^*}$ and $x$ is a variable in $q$, as follows: 
\begin{itemize}
\item the root node of $\mathcal{T}$ is $(c, y, \tau)$; and
\item whenever $(a, u, \sigma)$ is a node in $\mathcal{T}$ with a rooted tree set $\sigma$ starting in $a$ in $\rep^*$ for an atom $R(\underline{u}, u_1,u_2,\dots,u_n)$ in $q$ and fact $R(\underline{a}, b_1,b_2,\dots,b_n)$ in $\rep^*$, 

\begin{itemize}
\item[(i)] if $\streecfg{q}{y}$ invokes a forward production rule 
\begin{align*}
S_u &\rightarrow_q R(S_{u_1}, S_{u_2}, \dots, S_{u_n}),
\end{align*} 
then the node $(a, u, \sigma)$ has $n$ outgoing $R$-edges to its children $(b_1, {u_1}, \tau_1)$, $(b_2, {u_2}, \tau_2)$, $\dots$, $(b_n, {u_n}, \tau_n)$; or
\item[(ii)] if $\streecfg{q}{y}$ invokes a backward production rule $S_u \rightarrow_q S_v$, 
then the node $(a, u, \sigma)$ has a single outgoing $\emptyword$-edge to its only child~$(a, v, \sigma)$.
\end{itemize}
\end{itemize} 

The tree trace $\mathcal{T}$ succinctly records the rule productions that witness $\tau\in\streecfg{q}{y}$ in $\rep^*$.
We use a structural induction to show that for every node $(a, u, \sigma)$ in $\mathcal{T}$, $\pair{a}{u} \in B$.

\begin{itemize}
\item Basis. The claim holds for every leaf node $(a, u, \sigma)$ in $\mathcal{T}$, since if $\sigma=\bot$, then $\pair{a}{u}\in B$, or otherwise $\sigma=L$ starting in $a$ in $\rep^*$ for some unary relation name $L$, and we have $L(a)$ is in $\db$.

\item Inductive step. Let $(a, u, \sigma)$ be a node in $\mathcal{T}$. 
Assume that for any child node $(b, w, \sigma')$ of $(a, u)$ in $\mathcal{T}$ (possibly $b = a$), $\pair{b}{w} \in B.$
It suffices to argue that for the atom $\atomAt{R}{u}{q} = R(\underline{u}, u_1, u_2, \dots, u_n)$ in $q$, $\pair{a}{u} \in B$.

(i) Case that $(a, u, \sigma)$ has child nodes $(b_1,{u_1},\tau_1)$, $(b_2, {u_2},\tau_2)$, $\dots$, $(b_n, {u_n},\tau_n)$ in $\mathcal{T}$ with $\sigma=R(\tau_1,\tau_2,\dots,\tau_n)$.
By the inductive hypothesis
$\pair{b_i}{u_i} \in B$ for every $1 \leq i \leq n$,
which yields that $\factformula{R}{a}{\vec{b}}{u}$ is true, where $\vec{b} = \langle b_1, b_2, \dots, b_n \rangle$.
Then by Lemma~\ref{lemma:frugal-fact-choice-static}, $\pair{a}{u} \in B$.

(ii) Case that $(a, u, \sigma)$ has a child node $(a, v, \sigma)$ in $\mathcal{T}$ connected with an $\emptyword$-edge.
Then there is some atom $\atomAt{R}{v}{q}$ with $\atomAt{R}{v}{q} \precedes{q} \atomAt{R}{u}{q}$. By the inductive hypothesis on the child $(a, v, \sigma)$, $\pair{a}{v} \in B$. Hence $\pair{a}{u} \in B$ by Lemma~\ref{lemma:chain-property}.
\end{itemize}
This completes the proof.
\end{proof}

The proof of Lemma~\ref{lemma:correctness} can now be given. 

\begin{proof}[Proof of Lemma~\ref{lemma:correctness}]
\framebox{\ref{it:semantics}$\implies$\ref{it:syntax}}
Let $\rep^*$ be a frugal repair of $\db$ with respect to $q$. Then there is a rooted tree set $\tau$ starting in $c$ in $\rep^*$ with $\tau \in \streecfg{q}{y}$. The claim follows by Lemma~\ref{lemma:frugal-no-tree}.

\framebox{\ref{it:syntax}$\implies$\ref{it:semantics}}
Assume that $\pair{c}{y} \in B$. We use induction on $k$ to show that if $\pair{c}{y}$ is added to $B$ at the $k$-th iteration, then for any repair $\rep$ of $\db$, there exists a rooted tree set $\tau$ starting in $c$ in $\tau$ with $\tau \in \streecfg{q}{y}$. 

\begin{itemize}
\item Basis $k = 0$. Then every $\pair{c}{u}$ is added to $B$ for every leaf variable $u$ of $q$ such that either the label of $u$ in $q$ is $\bot$, or a unary relation name $L$. If the label of $u$ is $\bot$, the empty rooted tree set $\tau=\emptyset$ starting in $c$ with string representation $\bot$ is accepted by $\streecfg{q}{u}$, Otherwise, we must have $L(c) \in \db$, and the rooted tree set $\tau=L$ starting in $c$ is accepted by $\streecfg{q}{u}$. 

\item Inductive step.
Assume that $\pair{c}{y}$ is added to $B$ in the $k$-th iteration, and for any tuple $\pair{b}{x}$ added to $B$ prior to the addition of $\pair{c}{y}$, any repair of $\db$ contains a rooted tree set $\tau\in\streecfg{q}{x}$ starting in $b$. 
Let $\rep$ be any repair of $\db$. 
It suffices to construct a rooted tree set $\tau$ in $\rep$ starting in $c$ such that $\tau \in \streecfg{q}{y}$. 
Let $\atomAt{R}{y}{q} = R(\underline{y}, y_1, y_2,\dots,y_n)$.
Let $R(\underline{c}, d_1,d_2,\dots,d_n) \in \rep$ and let $\vec{d} = \langle d_1, d_2, \dots, d_n \rangle$.
Since $\pair{c}{y}\in B$, $\factformula{R}{c}{\vec{d}}{y}$ is true. Consider two cases.

\begin{itemize}
\item Case that $\pair{d_i}{y_i}\in B$ for every $1 \leq i \leq n$. Since each $\pair{d_i}{y_i}$ was added to $B$ in an iteration $<k$, by the inductive hypothesis, there is a rooted tree set $\tau_i$ starting in $d_i$ in $\rep$ with $\tau_i \in \streecfg{q}{y_i}$, i.e., $S_{y_i} \derive_q \tau_i$. 
Consider the rooted tree set $\tau = \{R(\underline{c}, \vec{d})\} \cup \bigcup_{1 \leq i \leq n} \tau_i$, starting in $c$ in $\rep$ with a string representation $\tau = R(\tau_1, \tau_2, \dots, \tau_n)$. From
\begin{align*}
S_y &\rightarrow_q R(S_{y_1}, S_{y_2}, \dots, S_{y_n}) \\
&\derive_q R(\tau_1, \tau_2,\dots,\tau_n) = \tau,
\end{align*} 
we conclude that $\tau \in \streecfg{q}{y}$.

\item Case that $\factformula{R}{c}{\vec{d}}{x}$ is true for some $\atomAt{R}{x}{q} \precedes{q} \atomAt{R}{y}{q}$.
Without loss of generality, we assume that $x$ is the smallest with respect to $\precedes{q}$ for the atom $R(\underline{x}, x_1,x_2,\dots,x_n)$.
Hence we must have $\pair{d_i}{x_i}\in B$ for every $1 \leq i \leq n$, and by the previous case, there exists a rooted tree set $\tau_i$ starting in $d_i$ such that $\tau_i \in \streecfg{q}{x_i}$, i.e., $S_{x_i} \derive_q \tau_i$.
Since $\atomAt{R}{x}{q} \precedes{q} \atomAt{R}{y}{q}$, we have 
\begin{align*}
S_y &\rightarrow_q S_x \\
&\rightarrow_q R(S_{x_1}, S_{x_2}, \dots, S_{x_n}) \\
&\derive_q R(\tau_1, \tau_2,\dots,\tau_n) = \tau,
\end{align*}
and therefore $\tau\in\streecfg{q}{y}$.
\end{itemize}
\end{itemize}
The proof is hence complete.  
\end{proof}

\subsection{Expressibility in \LFPL~and \FO}
\label{sec:membership}

\begin{lemma}
\label{prop:certaintrace-ptime}
For every query $q$ in $\rtbcq$ that satisfies $\cbranch$, $\certaintrace{q}$ is expressible in \LFPL\ (and thus is in~\PTIME).
\end{lemma}
\begin{proof}
Let $r$ be the root variable of $q$. 
Our algorithm first computes the set $B$, and then checks 
$\exists c: \pair{c}{r} \in B.$
The algorithm is correct by Lemma~\ref{lemma:correctness}.

For a rooted tree query~$q$, define the following formula in \LFPL~\cite{DBLP:books/sp/Libkin04}:
\begin{equation}\label{eq:lfp}
\psi_{q}(s,t)\defeq\left[\mathbf{lfp}_{B,x,z}\varphi_{q}(B,x,z)\right](s,t),
\end{equation}
where we have
\begin{equation*}
\varphi_{q}(B,x,z):=
\left(
\begin{array}{ll}
& \alpha(x) \land \formula{z = \atomAt{R}{u}{q}} \\[1ex]
\land & \exists y R(\underline{x}, \vec{y}) \\[1ex]
\land & \forall \vec{y} \formula{R(\underline{x}, \vec{y}) \rightarrow \factformula{R}{x}{\vec{y}}{u}}
\end{array}
\right)
\end{equation*}
and the formula $\factformula{R}{x}{\vec{y}}{u}$ is defined in Fig.~\ref{fig:algo}.
The initialization step of $B$ in Fig.~\ref{fig:algo} is expressible in \FO.
Herein, $\alpha(x)$ denotes a first-order query that computes the active domain.
That is, for every database instance $\db$ and constant $c$,
$\db\models\alpha(c)$ if and only if $c\in\adom{\db}$.
It is easy to verify that the LFP formula in~(\ref{eq:lfp}) computes the set $B$ in Fig.~\ref{fig:algo}.
\end{proof}

We now show that if $q$ satisfies $\cone$, we can safely remove the recursion from the algorithm in Fig.~\ref{fig:algo}.

\begin{lemma}
\label{lemma:fact-forward-suffices}
Let $q$ be a rooted tree query satisfying $\cone$. 
Let $\db$ be a database. 
Let $\atomAt{R}{y}{q} = R(\underline{y}, y_1, y_2, \dots, y_n)$ be an atom in $q$ and let $R(\underline{c}, \vec{d}) = R(\underline{c}, d_1, d_2,\dots, d_n)$ be a fact in $\db$.
Then $\factformula{R}{c}{\vec{d}}{y}$ is true if and only if for every atom $\atomAt{T_i}{y_i}{q}$ in $q$, $\pair{d_i}{y_i} \in B$.
\end{lemma}
\begin{proof}
\framebox{$\impliedby$}
Immediate by definition of $\factformula{R}{c}{\vec{d}}{y}$. 

\framebox{$\implies$}
Assume that $\factformula{R}{c}{\vec{d}}{y}$ is true. 
Let $\atomAt{R}{x}{q}$ be a minimal atom with respect to $\precedes{q}$ such that $\atomAt{R}{x}{q} \precedes{q} \atomAt{R}{y}{q}$ and $\factformula{R}{c}{\vec{d}}{x}$ is true.
If such an atom $\atomAt{R}{x}{q}$ does not exist, then the claim follows by definition of $\factformula{R}{c}{\vec{d}}{y}$. Otherwise, since $\atomAt{R}{x}{q}$ is minimal with respect to $\precedes{q}$, 
for every atom $\atomAt{T_i}{x_i}{q}$ in $q$, $\pair{d_i}{x_i} \in B$, where $R(\underline{x}, \vec{x}) = R(\underline{x}, x_1, x_2, \dots, x_n)$.

It suffices to show that for every atom $\atomAt{T_i}{y_i}{q}$ in $q$, $\pair{d_i}{y_i} \in B$.
Let $\atomAt{T_i}{y_i}{q}$ be an atom in $q$.
From $\cone$ and $\atomAt{R}{x}{q} \precedes{q} \atomAt{R}{y}{q}$, $\rootedAt{y_i}{q} \leqhomo{y_i}{x_i} \rootedAt{x_i}{q}$. Thus there is some atom $\atomAt{T_i}{x_i}{q}$ in $q$ with $\atomAt{T_i}{x_i}{q} \precedes{q} \atomAt{T_i}{y_i}{q}$.
Since $\pair{d_i}{x_i} \in B$, by Lemma~\ref{lemma:formula-via-homo}, $\pair{d_i}{y_i} \in B$. 
\end{proof}

\begin{lemma}
\label{coro:cone-tree-trace-fo}
For every $q$ in $\rtbcq$ that satisfies $\cone$, $\certaintrace{q}$ is in \FO\ .
\end{lemma}
\begin{proof}
Consider the following variant of the algorithm in Fig.~\ref{fig:algo}, where we simply have 
$$\factformula{R}{c}{\vec{d}}{y} = \bigwedge_{1 \leq i \leq n} \pair{d_i}{y_i} \in B.$$
The variant algorithm is correct for $\certaintrace{q}$ by Lemma~\ref{lemma:fact-forward-suffices}. 
Since the size of the query $q$ is fixed, for every constant $c$ and variable $y$ in $q$, deciding whether $\pair{c}{y} \in B$ is in \FO~since the algorithm in Fig.~\ref{fig:algo} can be expanded into a sentence of fixed size.
So is our algorithm, which checks
$\exists c: \pair{c}{r} \in B.$
\end{proof}

   





\section{Complexity Upper Bounds}
\label{sec:algorithm}

In this section, we prove the upper bound results in Theorem~\ref{thm:trichotomy}. 
First, we shall prove Lemma~\ref{lem:min-start}.


\begin{lemma}
\label{lemma:factor-of-treecfg}
Let $q$ be a rooted tree query. Then $q$ satisfies $\cfactor$ if and only if $q \leqhomo{}{} \tau$ for every $\tau \in \treecfg{q}$.
\end{lemma}
\begin{proof}[Proof of Lemma~\ref{lemma:factor-of-treecfg}]
Consider two directions.

\framebox{$\impliedby$} Let $\atomAt{R}{x}{q}$ and $\atomAt{R}{y}{q}$ be two atoms in $q$ with  $\atomAt{R}{x}{q} \precedes{q} \atomAt{R}{y}{q}$. 
It suffices to show that $\rewind{q}{y}{x}{R} \in \treecfg{q}$. 
Indeed, there is an execution of $S_r(\rewind{q}{y}{x}{R})$ that follows exactly $S_r(q)$, until it invokes $S_y(\rootedAt{x}{q})$, instead of $S_y(\rootedAt{y}{q})$ in $S_r(q)$. 
Note that $S_y \rightarrow_q S_x \derive_q \rootedAt{x}{q}$.
Thus $S_r \derive_q \rewind{q}{y}{x}{R}$, concluding that $\rewind{q}{y}{x}{R} \in \treecfg{q}$.

\framebox{$\implies$} 
Let $\tau \in \treecfg{q}$ with $S_r \derive_q \tau$. We use an induction on the number $k$ of backward transitions in $S_r \derive_q \tau$ to show that $q\leqhomo{}{} \tau$.

\begin{itemize}
\item Basis $k = 0$. We have $\tau = q$, and the claim follows.
\item Inductive step $k \rightarrow k+1$. Assume that if $S_r \derive_q \sigma$ uses $k$ backward transitions, then $q \leqhomo{}{} \sigma$. 

Let $\tau\in\treecfg{q}$ such that $S_r \derive_q \tau$ uses $k+1$ backward transitions. Let $\sigma$ be a subtree of $\tau$ such that the execution of $S_r(\sigma)$ invokes exactly $1$ backward transition $S_y \rightarrow_q S_x \derive_q \sigma$. Hence $\sigma=\rootedAt{x}{q}$. Consider the rooted tree $\tau^*$, obtained by replacing $\sigma=\rootedAt{x}{q}$ with $\sigma^* = \rootedAt{y}{q}$. We have $\tau^* \in \treecfg{q}$, since $S_r \derive_q \tau$ would invoke $S_y \derive_q \sigma^*$ and use exactly $k$ backward transitions. By the inductive hypothesis, there is a homomorphism $h$ from $q$ to $\tau^*$. 
If $h(q) \cap \sigma^* = \emptyset$, then $h(q)$ is still present in $\tau$, and thus $q \leqhomo{}{} \tau$. Otherwise, assume that the homomorphism $h$ maps $\rootedAt{z}{q}$ in $q$ to $\sigma^*$. Hence $\atomAt{R}{x}{q} \precedes{q} \atomAt{R}{y}{q} \precedes{q} \atomAt{R}{z}{q}$. Since $q$ satisfies $\cfactor$, there is a homomorphism $g$ from $q$ to $\rewind{q}{z}{x}{R}$, and thus a homomorphism from $q$ to $\tau$.
\end{itemize}
The proof is now complete.
\end{proof}

The following definition is helpful in our exposition.

\begin{definition}
Let $q$ be a rooted tree query. 
Let $\db$ be a database. For each repair $\rep$ of $\db$, we define $\starttwo{q}{\rep}$ as the set containing all (and only) constants $c \in \adom{\rep}$ such that there is a rooted tree set $\tau$ in $\rep$ starting in $c$ with $\tau \in \treecfg{q}$.
\end{definition}

The problem $\certaintrace{q}$ essentially asks whether there is some constant $c$ such that for every repair $\rep$ of $\db$, $c\in\starttwo{q}{r}$. 
Surprisingly, the frugal repair $\rep^*$ of $\db$ minimizes $\starttwo{q}{\cdot}$ across all repairs of $\db$.

\begin{lemma}
\label{lem:key}
Let $q$ be a rooted tree query satisfying $\cbranch$. Let $\db$ be a database. Let $\rep^*$ be a frugal repair of $\db$. Then for any repair $\rep$ of $\db$, $\starttwo{q}{\rep^*} \subseteq \starttwo{q}{\rep}$.
\end{lemma}
\begin{proof}[Proof of Lemma~\ref{lem:key}]
Let $B$ be the output of the algorithm in Fig.~\ref{fig:algo}. 
Let $\rep^*$ be a frugal repair of $\db$.
Let $\rep$ be any repair of $\db$. 
We show that $\starttwo{q}{\rep^*} \subseteq \starttwo{q}{\rep}$. 
Let $r$ be the root variable of $q$.
Assume that $c \in \starttwo{q}{\rep^*}$. Then there exists a rooted tree set $\tau$ starting in $c$ in $\rep^*$ with $\tau \in  \treecfg{q} = \streecfg{q}{r}$. By Lemma~\ref{lemma:frugal-no-tree}, we have $\pair{c}{r} \in B$.
By Lemma~\ref{lemma:correctness}, there exists a rooted tree set $\tau'$ starting in $c$ in $\rep$ with $\tau' \in \streecfg{q}{r} =\treecfg{q}$. Thus $c \in \starttwo{q}{\rep}$.
\end{proof}

The proof of Lemma~\ref{lem:min-start} can now be given.

\begin{proof}[Proof of Lemma~\ref{lem:min-start}]
\framebox{\ref{it:rifidb}$\implies$\ref{it:rificonstant}}
Assume~\eqref{it:rifidb}.
Let $\rep^*$ be a frugal repair of $\db$.
Since $\rep^{*}$ satisfies $q$, there is a rooted tree set starting in $c$ isomorphic to $q$ in $\rep^{*}$.
Since~$q\in\treecfg{q}$, we have $c\in\starttwo{q}{\rep^{*}}$.
By Lemma~\ref{lem:key}, for every repair $\rep$ of $\db$, $\starttwo{q}{\rep^{*}}\subseteq\starttwo{q}{\rep}$.
It follows that $c\in\starttwo{q}{\rep}$ for every repair $\rep$ of $\db$.

\framebox{\ref{it:rificonstant}$\implies$\ref{it:rifidb}}
Let $\rep$ be any repair of $\db$. 
By our hypothesis that~\eqref{it:rificonstant} holds true, there is some $c\in\starttwo{q}{\rep}$.
Let $\tau$ be a rooted tree set in $\rep$ starting in $c$ with $\tau \in \treecfg{q}$.
Since $q$ satisfies $\ctwo$ by the hypothesis of the current lemma, it follows by Lemma~\ref{lemma:factor-of-treecfg} that $q \leqhomo{}{} \tau$.
Consequently, $\rep$ satisfies~$q$.
\end{proof}

The upper bounds in Theorem~\ref{thm:trichotomy} thus follows.

\begin{proposition}
\label{coro:ctwo}
For every $q$ in $\rtbcq$,
\begin{enumerate}
\item if $q$ satisfies $\ctwo$, then $\cqa{q}$ is in \LFPL; and
\item if $q$ satisfies $\cone$, then $\cqa{q}$ is in \FO.
\end{enumerate}
\end{proposition}
\begin{proof}
Immediate from Lemmas~\ref{lem:min-start}, \ref{prop:certaintrace-ptime}, and~\ref{coro:cone-tree-trace-fo} by noting that $\cone$ implies~$\ctwo$. 
\end{proof}


Interestingly, for each query $q$ in $\rtbcq$ satisfying $\ctwo$, ``checking the frugal repair is all you need''.
A repair with this property is known as a ``universal repair'' in~\cite{DBLP:conf/icdt/CateFK12}.


\begin{corollary}
\label{prop:frugal-perspective}
Let $q$ be a query in $\rtbcq$, and let $\db$ be a database instance. 
Let $\rep^*$ be a frugal repair of $\db$ with respect to $q$. 
If~$q$ satisfies $\ctwo$, then $\db$ is a ``yes''-instance of $\cqa{q}$ if and only if $\rep^*$ satisfies $q$.
\end{corollary}
\begin{proof}
Let $\rep^*$ be a frugal repair of $\db$ with respect to $q$.

\framebox{$\implies$} This direction is straightforward. 
\framebox{$\impliedby$} Assume that $\rep^*$ satisfies $q$. 
Let $r$ be the root variable of $q$.
Hence there exists a constant $c$ in $\db$, such that there exists a rooted relation tree $\tau$ in $\rep^*$ that is isomorphic to $q$ accepted by $\streecfg{q}{r}$.
Then by Lemma~\ref{lemma:frugal-no-tree}, $\pair{c}{r}\in B$, where $B$ is the output of algorithm in Fig.~\ref{fig:algo}.
Hence $\db$ is a ``yes''-instance for $\certaintrace{q}$, and by Lemma~\ref{lem:min-start}, a ``yes''-instance for $\cqa{q}$.
\end{proof}


\section{Complexity Lower Bounds}
\label{sec:hardness}

In this section, we show the hardness results in Theorem~\ref{thm:trichotomy}.

We define a {\em canonical copy} of a query $q$ as a set of facts $\mu(q)$, where $\mu$ maps each variable in $q$ to a unique constant. 
The following notation will be central in all our reductions. For a query $q$, variables $x_i$ in $q$ and distinct constants $c_i$, we denote 
$$\qcopy{q}{x_1, x_2, \dots, x_n}{c_1, c_2, \dots, c_n}$$ as the canonical copy $\mu(q)$, where 
$$\mu(z) = \begin{cases} c_i & \text{if } z = x_i \\ \text{a fresh distinct constant} & \text{otherwise}. \end{cases}$$

\begin{lemma}
\label{lemma:cbranch-conp-hard}
$\cqa{q}$ is \coNP-hard for each $q$ in $\rtbcq$ that violates $\ctwo$.
\end{lemma}
\begin{proof}
Since $q$ violates $\ctwo$, there exist  two atoms $R(\underline{p}, \dots)$ and $R(\underline{n}, \dots)$ in $q$ such that there is no homomorphism from $q$ to neither $\rewind{q}{p}{n}{R}$ nor $\rewind{q}{n}{p}{R}$. 

Consider now the root atom $A(\underline{r}, \dots)$. It must be that $r \neq p$, since otherwise, there would be a homomorphism from $q$ to $\rewind{q}{n}{p}{R}$, a contradiction. Similarly, we have that $r \neq n$.  Hence, the root atom is distinct from $R(\underline{p}, \dots)$ and $R(\underline{n}, \dots)$.
We also have that $r \precedes{q} p$ and   $r \precedes{q}  n$.

We present a reduction from $\mathsf{Monotone SAT}$: Given a monotone CNF formula $\varphi$, i.e.,~each clause in $\varphi$ contains either all positive literals or all negative literals, does $\varphi$ has a satisfying assignment? 

Let $\varphi$ be a monotone CNF formula.
We construct an instance $\db$ for $\cqa{q}$ as follows. 
\begin{itemize}
    \item for each variable $z$ in $\varphi$, we introduce the facts $\qcopy{\rootedAt{p}{q}}{p}{z}$ and $\qcopy{\rootedAt{n}{q}}{n}{z}$;
    \item for each positive literal $z$ in clause $C$, we introduce the facts $\qcopy{q \setminus {\rootedAt{p}{q}}}{r, p}{C, z}$;
   \item for each negative literal $z$ in clause $\overline{C}$, we introduce the facts $\qcopy{q \setminus {\rootedAt{n}{q}}}{r, n}{\overline{C}, z}$;
\end{itemize}

Observe that the instance $\db$ has two types of inconsistent blocks. For relation $A$, we have a block for each positive or negative clause, where the primary key position is the clause. For relation $R$, for every variable $z$ we have a block of size two, which corresponds to choosing a true/false assignment for $z$. All the other relations are consistent.

Additionally, for a positive literal $z \in C$, the set of facts $\qcopy{\rootedAt{p}{q}}{p}{z} \cup \qcopy{q \setminus {\rootedAt{p}{q}}}{r, p}{C, z}$ make $q$ true; similarly for a negative literal $z \in \overline{C}$, the facts $\qcopy{\rootedAt{n}{q}}{n}{z} \cup \qcopy{q \setminus {\rootedAt{n}{q}}}{n, p}{\overline{C}, z}$ make $q$ true. Note also that $\qcopy{\rootedAt{n}{q}}{n}{z} \cup \qcopy{q \setminus {\rootedAt{p}{q}}}{r, p}{C, z}$ is a canonical copy of $\rewind{q}{p}{n}{R}$ (and hence cannot satisfy $q$), while $\qcopy{\rootedAt{p}{q}}{p}{z} \cup \qcopy{q \setminus {\rootedAt{n}{q}}}{r, n}{\overline{C}, z}$ is a canonical copy of $\rewind{q}{n}{p}{R}$ (which also cannot satisfy $q$).

Now we argue that $\varphi$ has a satisfying assignment $\chi$ if and only if $\db$ has a repair $\rep$ that does not satisfy $q$.

\framebox{$\implies$}
Assume that $\varphi$ has a satisfying assignment $\chi$. Consider the repair $\rep$ of $\db$ that
\begin{itemize}
\item for each variable $z$, if $\chi(z) = \mathsf{true}$, picks $\qcopy{\rootedAt{n}{q}}{n}{z}$, or otherwise $\qcopy{\rootedAt{p}{q}}{p}{z}$;
\item for each positive clause $C$, picks $\qcopy{q \setminus {\rootedAt{p}{q}}}{r, p}{C, z}$ where $z$ is a positive literal in $C$ with $\chi(z) =  \mathsf{true}$; and
\item for each negative clause $\overline{C}$, picks $\qcopy{q \setminus {\rootedAt{n}{q}}}{r, n}{\overline{C}, \overline{z}}$ where $\overline{z}$ is a negative literal in $\overline{C}$ with $\chi(z) =  \mathsf{false}$.
\end{itemize}


We show that $\rep$ does not satisfy $q$. Indeed, for each positive clause $C$, there is a literal $z \in C$ with $\chi(z) =  \mathsf{true}$, and thus $\qcopy{q \setminus {\rootedAt{p}{q}}}{r, p}{C, z} \subseteq \rep$. 
However, we have $\qcopy{\rootedAt{n}{q}}{n}{z} \subseteq \rep$, and thus $q$ is not satisfied.
Similarly, for each negative clause $\overline{C}$, there is a literal $\overline{z} \in \overline{C}$ with $\chi(z) = \mathsf{false}$, and thus $\qcopy{q \setminus {\rootedAt{n}{q}}}{n, p}{\overline{C}, z} \subseteq \rep$. 
However, we have $\qcopy{\rootedAt{p}{q}}{p}{z} \subseteq \rep$ and hence this part also cannot satisfy $q$.
Hence $\rep$ does not satisfy $q$.

\framebox{$\impliedby$}
Now assume that $\db$ has a repair $\rep$ that does not satisfy $q$. 
Consider the assignment $\chi$ that sets $\chi(z) = \mathsf{true}$ if $\qcopy{\rootedAt{n}{q}}{n}{z} \subseteq \rep$, or otherwise $\chi(z) = \mathsf{false}$. 
We argue that $\varphi$ is satisfied. For each positive clause $C$, there exists some $z \in C$ such that $\qcopy{q \setminus {\rootedAt{p}{q}}}{r, p}{C, z} \subseteq \rep$. Since $\rep$ does not satisfy $q$, it must be that $\qcopy{\rootedAt{p}{q}}{p}{z} \nsubseteq \rep$ and thus $\qcopy{\rootedAt{n}{q}}{n}{z} \subseteq \rep$. By construction, $z$ is true and the clause $C$ is satisfied. Similarly, the negative clauses are all satisfied by the assignment.
\end{proof}

\begin{lemma}
\label{lemma:nl-hardness-subtree}
Let $q$ be a rooted tree query. If there exist two distinct atoms $R(\underline{x}, \dots)$ and $R(\underline{y}, \dots)$ such that $x \precedes{q} y$ and there is no root homomorphism from $\rootedAt{y}{q}$ to $\rootedAt{x}{q}$ (i.e., it does not hold that $\rootedAt{y}{q} \leqhomo{y}{x} \rootedAt{x}{q}$), then $\cqa{q}$ is \NL-hard. 
\end{lemma}

\begin{proof}

We may assume without loss of generality two things $(i)$ there is no atom $R(\underline{z}, \dots)$ such that $z \notin \{x,y\}$, $x \precedes{q} z \precedes{q} y$ (we then say that the two $R$-atoms are consecutive), and $(ii) $ for any $y \precedes{q} z$, $z \neq y$, we have $\rootedAt{z}{q} \leqhomo{z}{y} \rootedAt{y}{q}$. Indeed, we can pick $R(\underline{x}, \dots)$ and $R(\underline{y}, \dots)$ to be the pair of consecutive $R$-atoms that violates the root homomorphism condition and occurs lowest in the rooted tree. Such a pair must always exists, since the root homomorphism property is transitive, i.e., if $\rootedAt{y}{q} \leqhomo{y}{z} \rootedAt{z}{q}$ and $ \rootedAt{z}{q} \leqhomo{z}{x} \rootedAt{x}{q}$, then  we also have that $\rootedAt{y}{q} \leqhomo{y}{x} \rootedAt{x}{q}$.


We present a reduction from the complement of \textsf{REACHABILITY} problem, which is \NL-hard: Given a directed acyclic graph $G = (V, E)$ and $s, t \in V$, is there a directed path from $s$ to $t$ in $G$?

We construct an instance $\db$ for $\cqa{q}$ as follows. First, we introduce two new constants $s'$ and $t'$. Then:
\begin{itemize}
\item for each $u \in V \cup \{s'\}$, introduce $\qcopy{q \setminus \rootedAt{x}{q}}{x}{u}$;
\item for every edge $(u, v) \in E \cup \{(s', s), (t, t')\}$, introduce $\qcopy{\rootedAt{x}{q} \setminus \rootedAt{y}{q}}{x, y}{u, v}$;
\item for every vertex $u \in V$, introduce $\qcopy{\rootedAt{y}{q}}{y}{u}$.
\end{itemize}

Note that the above construction guarantees that only $R$ has inconsistent blocks.

We now argue that there is a directed path $(u_1, u_2, \dots, u_k)$ with $(u_i, u_{i+1}) \in E$, $u_1 = s$ and $u_k = t$ in $G$ if and only if there is a repair of $\db$ that does not satisfy $q$.

\framebox{$\implies$}
Assume that there exists a directed path $(u_1, u_2, \dots, u_k)$ with $(u_i, u_{i+1}) \in E$, $u_1 = s$ and $u_k = t$ in $G$. 
Denote $u_0 = s'$ and $u_{k+1} = t'$.
Let $\rep$ be the repair that picks $\qcopy{\rootedAt{x}{q} \setminus \rootedAt{y}{q}}{x, y}{u_i, u_{i+1}}$ for every $1 \leq i \leq k-1$, and $\qcopy{\rootedAt{y}{q}}{y}{u}$ for any other vertex $u$. Suppose for contradiction that $\rep$ satisfies $q$ with a valuation $\theta$. 
It is not possible that $\theta(q) \subseteq \qcopy{\rootedAt{y}{q}}{y}{u} $ for any $u \notin V$ since the size does not fit. 

We argue that we must have $\theta(x) = u_i$ and $\theta(y) = u_{i+1}$ for some $0 \leq i < k$. 
If $\theta(x) = u_i \in \{u_0, u_1, \dots,u_k\}$, then we must have $\theta(y) = u_{i+1}$ since $\qcopy{\rootedAt{x}{q} \setminus \rootedAt{y}{q}}{x, y}{u, v}$ is a canonical copy.
Suppose for contradiction that $\theta(x) \notin \{u_0, u_1, \dots, u_k\}$. 
It is not possible that $\theta(x) = u_{k+1} = t'$ since by construction, there is no rooted tree set rooted at $t'$.
Note that there is no atom $R(\underline{z}, \dots)$ such that $z \notin {x, y}$, $x \precedes{q} z \precedes{q} y$.
Hence $\theta(x)$ cannot fall on the path connecting any $u_{i}$ and $u_{i+1}$, and
$\theta(\rootedAt{x}{q})$ must be contained in some $\qcopy{\rootedAt{x}{q} \setminus \rootedAt{y}{q}}{x, y}{u_i, u_{i+1}}$.
Then, there must be an atom $R(\underline{z}, \dots)$ such that (i) $x \precedes{q} z$,
(ii) $z \vardisjoint{q} y$ and (iii) $\theta(\rootedAt{x}{q})$ is contained in $\qcopy{\rootedAt{z}{q}}{z}{\theta(x)}$, which is impossible since the sizes do not fit.

By construction, there is a canonical copy of $\rootedAt{y}{q}$ rooted at $u_{i+1}$. If this canonical copy is contained in  $\qcopy{\rootedAt{x}{q} \setminus \rootedAt{y}{q}}{x, y}{u_{i+1}, u_{i+2}}$, then there is a root homomorphism from  $\rootedAt{y}{q}$ to $\rootedAt{x}{q} \setminus \rootedAt{y}{q}$, and so  from  $\rootedAt{y}{q}$ to $\rootedAt{x}{q}$, a contradiction. Otherwise, there exists some atom $R(\underline{z}, \dots)$ such that $(i)$ $y \precedes{q} z$ and $(ii)$ $\rootedAt{y}{q} \setminus \rootedAt{z}{q}$ has a root homomorphism to $\rootedAt{x}{q} \setminus \rootedAt{y}{q}$. Recall now that from our initial assumption we must have that  $\rootedAt{z}{q} \leqhomo{z}{y} \rootedAt{y}{q}$. This implies that we can now generate a root homomorphism from $\rootedAt{y}{q}$ to $\rootedAt{x}{q}$, a contradiction.


\framebox{$\impliedby$}
Assume that there is no directed path from $s$ to $t$ in $G$. Consider any repair $\rep$ of $\db$. Since $G$ is acyclic, there exists a maximal sequence $u_0, u_1, \dots, u_k$ with $k\geq 1$ such that $u_0 = s'$, $u_1 = s$, $\qcopy{\rootedAt{x}{q} \setminus \rootedAt{y}{q}}{x, y}{u_i, u_{i+1}} \subseteq \rep$ for $0 \leq i < k$ and $\qcopy{\rootedAt{y}{q}}{y}{u_k} \subseteq \rep$. Then, the following set of facts satisfies $q$:
\begin{align*}
& \qcopy{q \setminus \rootedAt{x}{q}}{x}{u_{k-1}} \cup  \\
& \qcopy{\rootedAt{x}{q} \setminus \rootedAt{y}{q}}{x, y}{u_{k-1}, u_k} \cup \\
& \qcopy{\rootedAt{y}{q}}{y}{u_k}.
\end{align*}

This shows that $\cqa{q}$ is \NL-hard since \NL\ is closed under complement.
\end{proof}

\begin{lemma}
\label{lemma:nl-hardness}
$\cqa{q}$ is \NL-hard for each $q$ in $\rtbcq$ that violates $\cone$. 
\end{lemma}
\begin{proof}
Assume that $q$ violates $\cone$. 
Then there exist two distinct atoms $R(\underline{x}, \dots)$ and $R(\underline{y}, \dots)$ in $q$ such that there is no root homomorphism from $\rootedAt{y}{q}$ to $\rootedAt{x}{q}$ or from $\rootedAt{x}{q}$ to $\rootedAt{y}{q}$.
If $x \vardisjoint{q} y$,  Lemma~\ref{lemma:branch-to-root-homomorphism} implies that $\ctwo$ is also violated, so  $\cqa{q}$ is \coNP-hard from Lemma~\ref{lemma:cbranch-conp-hard}. Otherwise, $\cqa{q}$ is \NL-hard by Lemma~\ref{lemma:nl-hardness-subtree}.
\end{proof}

We conclude with the desired lower bounds.

\begin{proposition}
\label{prop:hardness}
For every $q$ in $\rtbcq$,
\begin{enumerate}
\item\label{it:conphardness} if $q$ violates $\ctwo$, then $\cqa{q}$ is \coNP-hard; and
\item if $q$ violates $\cone$, then $\cqa{q}$ is \NL-hard.
\end{enumerate}
\end{proposition}
\begin{proof}
Immediate from Lemma~\ref{lemma:cbranch-conp-hard} and~\ref{lemma:nl-hardness}.
\end{proof}

\section{Extending the Trichotomy}\label{sec:berge}

In this section, we extend the complexity classification for rooted tree queries to larger classes of Boolean conjunctive queries. 
We postpone most proofs to Appendix~\ref{appendix:berge}.

\subsection{From $\rtbcq$ to $\ibcq$}

We define $\ibcq$, a subclass of $\bcq$ that extends $\rtbcq$.

\begin{definition}[$\ibcq$]\label{def:ibcq}
$\ibcq$ is the class of Boolean conjunctive queries~$q$ satisfying the following conditions:
\begin{enumerate}
\item
every atom in $q$ is of the form $R(\underline{x},y_{1},\ldots,y_{n})$ where $x$ is a variable and $y_{1},\ldots,y_{n}$ are symbols (variables or constants) such that no variable occurs twice in the atom; and
\item
if $R(\underline{x},y_{1},\ldots,y_{n})$ and $S(\underline{u},v_{1},\ldots,v_{m})$ are distinct atoms of~$q$, then $x\neq u$.
Note that $R$ and $S$ need not be distinct.
\end{enumerate}
\end{definition}

For a query $q$ in $\bcq$, we define $\qgraph{q}$ as the undirected graph whose vertices are the atoms of~$q$; two atoms are adjacent if they have a variable in common. The connected components of~$q$ are the connected components of $\qgraph{q}$. Note that queries in $\ibcq$, unlike $\rtbcq$, can have more than one connected component. 
The following lemma implies that the complexity of $\cqa{q}$ is equal to the highest complexity of $\cqa{q'}$ over every connected component~$q'$ of~$q$.

\begin{lemma} \label{lemma:conn-hard}
    Let $q$ be a minimal query in $\bcq$ with connected components $q_1$, $q_2$, $\dots$, $q_n$. Then:
    \begin{enumerate}
  \item for every $1 \leq i \leq n$, there exists a first order reduction from the problem $\cqa{q_i}$ to $\cqa{q}$; and
  \item for every database instance $\db$, $\db$ is a ``yes''-instance of the problem $\cqa{q}$ if and only if for every $1 \leq i \leq n$, $\db$ is a ``yes''-instance of $\cqa{q_i}$.
  \end{enumerate}
\end{lemma}



\begin{proposition} \label{prop:berge}
If~$q$ is a connected minimal conjunctive query in $\ibcq\setminus\rtbcq$,  then $\cqa{q}$ is \LSPACE-hard (and not in \FO); if $q$ is also Berge-acyclic, then $\cqa{q}$ is \coNP-hard.
\end{proposition}


\begin{proof}[Proof of Theorems~\ref{thm:trichotomy:ext:no-berge} and~\ref{thm:trichotomy:ext}]
Let $q$ be a query in $\ibcq$. Then the minimal query $q^*$ of $q$ is also in $\ibcq$.
If every connected component of $q^*$ is in $\rtbcq$ and satisfies $\cone$, then $\cqa{q}$ is in \FO.
Otherwise, there exists some connected component $q'$ of $q^*$ that is either not in $\rtbcq$, or violates $\cone$,
and $\cqa{q}$ is \LSPACE-hard or \NL-hard by Lemma~\ref{lemma:conn-hard}, Proposition~\ref{prop:berge}, and Theorem~\ref{thm:trichotomy}. 
Assume that $q$ is also Berge-acyclic. If some connected component $q'$ of $q^*$ is not in $\rtbcq$, then $\cqa{q}$ is \coNP-complete; or otherwise, $\cqa{q}$ exhibits a trichotomy by Theorem~\ref{thm:trichotomy}.
\end{proof}

Lemma~\ref{lemma:bridging-lemma} (adapted from~\cite{DBLP:journals/corr/abs-1903-12469}) is essential to the proof of Proposition~\ref{prop:berge}, but is of independent interest.
It relates the complexity of CQA on queries with self-joins to that on self-join-free queries.

Given a query~$q$ in~$\bcq$, a \emph{self-join-free version of~$q$}, denoted $\sjf{q}$, is any self-join-free Boolean conjunctive query obtained from~$q$ by (only) renaming relation names. 
For example, a self-join-free version of $\{R(\underline{x},y), R(\underline{y},x)\}$ is $\{R(\underline{x},y), S(\underline{y},x)\}$.


\begin{lemma}[Bridging Lemma]
Let $q$ be a minimal query in $\bcq$ and $\mathcal{C}$ a complexity class. If $\cqa{\sjf{q}}$ is $\mathcal{C}$-hard, then $\cqa{q}$ is $\mathcal{C}$-hard. 
\label{lemma:bridging-lemma}
\end{lemma}

The Bridging Lemma is illustrated by Example~\ref{ex:bridging-lemma}.

\begin{example}
\label{ex:bridging-lemma}
For $q_1 = \{R(\underline{x}, y, z), R(\underline{z}, x, y)\}$, we have 
$\sjf{q_1} = \{R_1(\underline{x}, y, z), R_2(\underline{z}, x, y)\}$.
By Theorem~\ref{thm:sjf-theorem}~\cite{KoutrisW15}, $\cqa{\sjf{q_1}}$ is \LSPACE-complete, and thus $\cqa{q_1}$ is \LSPACE-hard by Lemma~\ref{lemma:bridging-lemma}.

For $q_2 = \{R(\underline{x}, z), R(\underline{y}, z)\}$, we have $\sjf{q_2} = \{R_1(\underline{x}, z), R_2(\underline{y}, z)\}$. 
Although by Theorem~\ref{thm:sjf-theorem}~\cite{KoutrisW15}, $\cqa{\sjf{q_2}}$ is \coNP-complete, $\cqa{q_2}$ is in \FO~because $q_2 \equiv q_2'$ where $q_2' = \{R(\underline{x}, z)\}$. 
Lemma~\ref{lemma:bridging-lemma} does not apply here because $q_2$ is not minimal.
\end{example}

\subsection{Open Challenges}\label{sec:beyond}

So far, we have established the \FO-boundary of $\cqa{q}$ for all queries~$q$ in $\ibcq$, and a fine-grained complexity classification for all Berge-acyclic queries in $\ibcq$, which include all rooted tree queries. 
We briefly discuss the remaining syntactic restrictions.

The complexity classification of $\cqa{q}$ for queries $q$ in $\ibcq$ that are not Berge-acyclic is likely to impose new challenges.
In particular, Figueira et al.~\cite{DBLP:conf/icdt/FigueiraPSS23} showed that for $q_1$ in Example~\ref{ex:bridging-lemma} (that is not Berge-acyclic), the complement of $\cqa{q_1}$ is complete for \textsf{Bipartite Matching} under \textsc{LOGSPACE}-reductions.


The restriction imposed by $\ibcq$ that every variable occurs at most once at a primary-key position allows for an elegant graph representation.
We found that dropping this requirement imposes serious challenges.
The following Proposition~\ref{prop:single-pk} hints at the difficulty  of having  to ``correlate two rooted tree branches'' that share the same primary-key variable.


\begin{proposition}
\label{prop:single-pk}
Consider the following queries:
\begin{itemize}
\item $q_1 = \{R(\underline{u}, x_1), R(\underline{x_1}, x_2), S(\underline{u}, y_1), S(\underline{y_1}, y_2)\};$
\item $q_2 = q_1 \cup \{X(\underline{x_2}, x_3)\};$ and
\item $q_3 = q_1 \cup \{X(\underline{x_2}, x_3), Y(\underline{y_2}, y_3)\}.$
\end{itemize}
Then we have $\cqa{q_1}$ is in \FO, $\cqa{q_2}$ is in \NL-hard $\cap$ \LFPL, and $\cqa{q_3}$ is \coNP-complete.
\end{proposition}

The restrictions that no atom contains repeated variables, and that no constant occurs at a primary-key position ease the technical treatment, but  it is likely that they can be dropped at the price of some technical involvement. On the other hand, all our techniques fundamentally rely on the restriction that primary keys are simple.

\section{Conclusion}
\label{sec:conclusion}
We established a fine-grained complexity classification of the problem $\cqa{q}$ for all rooted tree queries~$q$. 
We then lifted our complexity classification to a larger class of~queries.
A notorious open problem in consistent query answering is Conjecture~\ref{conj:dichotomy}, which conjectures that for every query $q$ in $\bcq$, $\cqa{q}$ is either in \PTIME\ or \coNP-complete. 
Despite our progress, this problem remains open even under the restriction that all primary keys are simple. 

\smallskip
\noindent \textbf{Acknowledgements.} The authors thank the anonymous reviewers for their constructive feedback and comments. 
This work is supported by the National Science Foundation under grant IIS-1910014 and the Anthony Klug NCR Fellowship.


  \bibliographystyle{abbrv}
  \bibliography{reference}
  \clearpage
\appendix

\section{Missing Proofs in Section~\ref{sec:trichotomy}}
\label{app:trichotomy}

\begin{proof}[Proof of Lemma~\ref{lemma:branch-to-root-homomorphism}]
We denote
$$p = \rewind{q}{y}{x}{R} = \formula{q \setminus \rootedAt{y}{q}} \cup f(\rootedAt{x}{q}),$$
for some isomorphism $f$ that maps every variable in $\rootedAt{x}{q}$ to a fresh variable, except for $x$, which we have $f(x) = y$.

Assume first that $\rootedAt{y}{q} \leqhomo{y}{x} \rootedAt{x}{q}$, witnessed by the homomorphism $h$ with $h(y) = x$. 
It is easy to verify that the homomorphism $g: \queryvars{q} \rightarrow \queryvars{p}$ with
$$
g(z) = \begin{cases}
z & \text{ if $z \in \queryvars{q \setminus \rootedAt{y}{q}}$}, \\
f(h(z)) & \text{ otherwise}
\end{cases}
$$
is a homomorphism from $q$ to $p$.

Assume there is a homomorphism $h: \queryvars{q} \rightarrow \queryvars{p}$ from $q$ to $p$. 
Hence 
$$h(q) = h(q \setminus \rootedAt{y}{q}) \cup h(\rootedAt{y}{q}) \subseteq \formula{q \setminus \rootedAt{y}{q}} \cup f(\rootedAt{x}{q}).$$
Note that $q$ is minimal, i.e., there is no automorphism $\alpha$ such that $\alpha(q) \subsetneq q$.
If $h(y) = y$, since $q$ is minimal, we have $h(q \setminus \rootedAt{y}{q}) = q \setminus \rootedAt{y}{q}$, and we have $h(\rootedAt{y}{q}) \subseteq f(\rootedAt{x}{q})$. Thus $\rootedAt{y}{q} \leqhomo{y}{x} \rootedAt{x}{q}$, witnessed by the homomorphism $g = f^{-1} \circ h $ with $g(y) = f^{-1}(h(y)) = f^{-1}(y) = x$, as desired.

Suppose for contradiction that $h(y) \neq y$. In this case, we have $h(q \setminus \rootedAt{y}{q}) \cap f(\rootedAt{x}{q}) \neq \emptyset$.

We argue that $y = f(x) \precedes{p} h(y)$. 
Case (i) Assume $h(y) \precedes{p} y=f(x)$ holds. 
Then $h$ maps the unique path of nodes from $r$ to $y$ in $q$
to the unique path from $h(r)$ to $h(y)$ in $p$. While we have $r = h(r)$ or $r \precedes{p} h(r)$,
but since $h(y) \precedes{p} y$, this is not possible because the path from $h(r)$ to $h(y)$ in $p$
is strictly shorter than the path from $r$ to $y$ in $q$.
Case (ii) Assume $h(y) \vardisjoint{p} y=f(x)$ holds. Let $y_0 = y$, and for each $i \geq 1$, $y_i = h(y_{i-1})$.
We argue that variables $y_0$, $y_1$, $\dots$ are all distinct, thereby
reaching a contradiction to the finite size of $q$.
Assume first that $y_1$ is a left sibling of $y_0$ in $q$: 
for the greatest common ancestor $y^*$ of $y_1$ and $y_0$, there
is an atom $R(\underline{y^*}, .., y_l, .., y_r)$ such that $y_l$ and $y_r$ are ancestors of
$y_1$ and $y_0$. The arguments for the case where $y_1$ is a right sibling of $y_0$
in $q$ is similar.
Note that $y_1$ appears in both $p$ and $q$ and its subtree is not affected by the
rewinding operation since $y_1 \vardisjoint{p} y_0$. Since $y_1$ is a left
sibling of $y_0$ and that the children of rooted trees are ordered,  
$h(y_1)$ is a left sibling of $h(y_0)$, that is $y_2$
is a left sibling of $y_1$ in $q$, and this process continues. Since each $y_
{i+1}$ is a left sibling of $y_{i}$, the variables need to be distinct, or
otherwise there is some $y_{j+1}$ is a right sibling of $y_{j}$, a contradiction.

Let $\atomAt{T}{z}{q}$ be the greatest common ancestor of $\atomAt{R}{x}{q}$ and $\atomAt{R}{y}{q}$ in $q$ and let $u$ and $v$ be variables in $\atomAt{T}{z}{q}$ such that $u \precedes{q} x$ and $v \precedes{q} y$ and $u \vardisjoint{q} v$. Hence $z$ appears in both $q$ and $p$. 
Hence $z \precedes{p} h(z)$ but $h(z) \neq z$. 
We have 
$$\card{\rootedAt{u}{q}} + \card{\rootedAt{v}{q}} + 1 \leq \card{\rootedAt{z}{q}} \leq \card{\rootedAt{h(z)}{p}},$$
because the homomorphism maps $\rootedAt{z}{q}$ to the subtree of $p$, rooted at $h(z)$.

We show that $v \precedes{q} h(z)$. Since $z \precedes{q} h(z)$ and $v$ is the immediate
child of $z$, we can have either $v \precedes{q} h(z)$ or $v \vardisjoint{q} h(z)$. Suppose for
contradiction that $v \vardisjoint{q} h(z)$, then $h(z) \notin \{u, v\}$. Then, $\rootedAt{h(z)}{p} = \rootedAt{h(z)}{q}$ since the rewinding leaves $\rootedAt{h(z)}{q}$ intact. But that implies
$h(\rootedAt{z}{q}) \subseteq \rootedAt{h(z)}{q}$ with $z \precedes{q} h(z)$, a contradiction.

Therefore, we have 
$$\card{\rootedAt{u}{q}} + \card{\rootedAt{v}{q}} + 1 \leq \card{\rootedAt{h(z)}{p}} \leq \card{\rootedAt{v}{p}} \leq \card{\rootedAt{v}{q}} - \card{\rootedAt{y}{q}} + \card{\rootedAt{x}{q}},$$
where the second inequality follows by construction of rewinding that replaces
$\rootedAt{y}{q}$ with $\rootedAt{x}{q}$.

This yields
$$0 \leq \card{\rootedAt{u}{q}} - \card{\rootedAt{x}{q}} \leq - \card{\rootedAt{y}{q}} - 1 < 0,$$
a contradiction.
\end{proof}

\section{Missing Proofs in Section~\ref{sec:certaintrace-ptime}}\label{app:certaintrace-ptime}

We first show that the formula in Fig.~\ref{fig:algo} propagate on root homomorphism.

\begin{lemma}
\label{lemma:formula-via-homo}
Let $q$ be a rooted tree query and $\db$ a database instance. 
Then for constants $c, d_1,d_2,\dots,d_n\in\adom{\db}$ where $\vec{d} = \langle d_1, d_2,\dots, d_n \rangle$ and any two atoms $\atomAt{R}{x}{q}$ and $\atomAt{R}{y}{q}$ with $\rootedAt{y}{q} \leqhomo{y}{x} \rootedAt{x}{q}$, the following statements hold:
\begin{enumerate}
\item if $\factformula{R}{c}{\vec{d}}{x}$ is true, then $\factformula{R}{c}{\vec{d}}{y}$ is true; and
\item if $\pair{c}{x} \in B$, then $\pair{c}{y} \in B$.
\end{enumerate}
\end{lemma}
\begin{proof}
We show both (1) and (2) by an induction on the height $k$ of the atom $\atomAt{R}{y}{q}$ in $q$.

\begin{itemize}
\item Basis $k = 0$. In this case, $y$ is a leaf variable of $q$ and (1) holds vacuously.
Assume that the label of $y$ is $L$, then there is an atom $L(\underline{y})$ in $q$. Then there must be an atom $L(\underline{x})$ in $q$. 
From $\pair{c}{x} \in B$, we have $L(\underline{c}) \in \db$, and thus $\pair{c}{y} \in B$ by the initialization step.

\item Inductive step. Assume that both (1) and (2) holds if the height of $\rootedAt{y}{q}$ is less than $k$. Consider the case where the height of $\rootedAt{y}{q}$ is $k$. 

First we show (1) holds. Assume that $\factformula{R}{c}{\vec{d}}{x}$ holds. Let $\atomAt{R}{x}{q} = R(\underline{x}, x_1, x_2, \dots, x_n)$ and $\atomAt{R}{y}{q} = R(\underline{y}, y_1, y_2, \dots, y_n)$. 
Consider two cases.

\begin{itemize}
\item Case (I) that the following formula is true
\begin{equation}
\label{eq:forward-holds}
\bigwedge_{1 \leq i \leq n} \pair{d_i}{x_i} \in B.
\end{equation}
To show $\factformula{R}{c}{\vec{d}}{y}$ holds, it suffices to show 
$$\bigwedge_{1 \leq i \leq n} \pair{d_i}{y_i} \in B.$$
Consider any $y_i$. If $y_i$ is a leaf variable with label $\bot$, then $\pair{d_i}{y_i} \in B$ by the initialization step.
Otherwise, there is an atom $\atomAt{T}{y_i}{q}$ in $q$.
Since $\rootedAt{y}{q} \leqhomo{y}{x} \rootedAt{x}{q}$, there is some atom $\atomAt{T}{x_i}{q}$ in $q$ such that $\rootedAt{y_i}{q} \leqhomo{y_i}{x_i} \rootedAt{x_i}{q}$ and $\pair{d_i}{x_i} \in B$, by Equation~(\ref{eq:forward-holds}). Since the height of $\atomAt{T}{y_i}{q}$ is less than $k$, by the inductive hypothesis for (2), we have $\pair{d_i}{y_i} \in B$.

\item Case (II) that there is some atom $\atomAt{R}{u}{q} \precedes{q} \atomAt{R}{x}{q}$, such that $\factformula{R}{c}{\vec{d}}{u}$ is true.

If $\atomAt{R}{u}{q} \precedes{q} \atomAt{R}{y}{q}$, then $\factformula{R}{c}{\vec{d}}{y}$ holds. Otherwise, we must have $\atomAt{R}{u}{q} \vardisjoint{q} \atomAt{R}{y}{q}$. Indeed, if not, we would have $\atomAt{R}{y}{q} \precedes{q} \atomAt{R}{u}{q} \precedes{q} \atomAt{R}{x}{q}$, but $\rootedAt{y}{q} \leqhomo{y}{x} \rootedAt{x}{q}$, a contradiction. 

We argue that $\rootedAt{y}{q} \leqhomo{y}{u} \rootedAt{u}{q}$. If not, then by $\cbranch$, we have $\rootedAt{u}{q} \leqhomo{u}{y} \rootedAt{y}{q} \leqhomo{y}{x} \rootedAt{x}{q}$, but $\atomAt{R}{u}{q} \precedes{q} \atomAt{R}{x}{q}$, a contradiction. 

Note that we just established $\factformula{R}{c}{\vec{d}}{u}$ is true and $\rootedAt{y}{q} \leqhomo{y}{u} \rootedAt{u}{q}$ for $\atomAt{R}{u}{q} \precedes{q} \atomAt{R}{x}{q}$. 
If Case (I) holds when $\factformula{R}{c}{\vec{d}}{u}$ is true, then $\factformula{R}{c}{\vec{d}}{y}$ is true, as desired. 
Otherwise, by the previous argument in Case (II), either $\factformula{R}{c}{\vec{d}}{y}$ is true as desired, or there is another atom $\atomAt{R}{w}{q}$ such that $\atomAt{R}{w}{q} \precedes{q} \atomAt{R}{u}{q} \precedes{q} \atomAt{R}{x}{q}$ and $\rootedAt{y}{q} \leqhomo{y}{w} \rootedAt{w}{q}$. Since there are only finitely many $R$-atoms in $q$, this process must terminate and show that $\factformula{R}{c}{\vec{d}}{y}$ is true.   
\end{itemize}

For (2), assume that $\pair{c}{x} \in B$. For every fact $R(\underline{c}, \vec{d}) \in \db$, $\factformula{R}{c}{\vec{d}}{x}$ holds. By (1), $\factformula{R}{c}{\vec{d}}{y}$ holds for every $R(\underline{c}, \vec{d}) \in \db$. Hence $\pair{c}{y} \in B$.
\end{itemize}

The proof is now complete.
\end{proof}

\begin{proof}[Proof of Lemma~\ref{lemma:chain-property}]
The lemma follows from Lemma~\ref{lemma:formula-via-homo} if $\rootedAt{y}{q} \leqhomo{y}{x} \rootedAt{x}{q}$. Assume that $\atomAt{R}{x}{q} \precedes{q} \atomAt{R}{y}{q}$, and both (1) and (2) are straightforward by definition of $\factformula{R}{c}{\vec{d}}{y}$.
\end{proof}

\section{Missing Proofs in Section~\ref{sec:berge}}
\label{appendix:berge}


\begin{proof}[Proof of Lemma~\ref{lemma:conn-hard}]
For item (1), let $\db$ be an instance for $\cqa{q_i}$ and construct an instance $\db' = \db \cup \bigcup_{j \neq i}D[q_j] $ for $\cqa{q}$, where $D[q_j]$ is a canonical copy of $q_j$. Clearly, $\db'$ can be constructed in \FO. Next, we show that $\cqa{q_i}$ is true on $\db$ if and only if $\cqa{q}$ is true on $\db'$.

Assume that $\db$ is a ``yes''-instance for $\cqa{q_i}$. Let $\rep'$ be any repair of $\db'$. 
We have $\bigcup_{j \neq i }D[q_j] \subseteq \rep'$, since each $D[q_j]$ is consistent. 
Thus $\rep = \rep' \setminus \bigcup_{j \neq i }D[q_j]$ is a repair of $\db$, and we have that $\rep$ satisfies $q_i$. Then $\rep'$ satisfies $q$, since $\rep$ also contains $D[q_j]$, which satisfies $q_j$ for $j \neq i$. 

Assume that $\db'$ is a ``yes''-instance for $\cqa{q}$. 
Let $\rep$ be any repair of $\db$. It remains to show that $\rep$ satisfy $q_i$.
Let $\rep' = \rep \cup \bigcup_{j \neq i} D[q_j]$. 
Hence $\rep'$ is a repair of $\db'$, and and there is a valuation $\mu$ such that $\mu(q) \subseteq \rep'$.
It suffices to show that $\mu(q_i) \subseteq \rep$. Suppose not, since each $q_i$ is connected and each $D[q_j]$ is connected, we must have $\mu(q_i) \subseteq D[q_j]$ for some $j \neq i$, which implies a homomorphism from $q_i$ to $q_j$, contradicting that $q$ is minimal. 
Thus $\mu(q_i) \subseteq \rep$, as desired.
\smallskip

Item (2) is proved in Lemma B.1 in~\cite{DBLP:conf/pods/KoutrisOW21}.	
\end{proof}

\begin{proof}[Proof of Bridging Lemma]
For each atom $N(\underline{\vec{x}}, \vec{y})$ in $\sjf{q}$, we denote $\pi(N) = R$ if $R(\underline{\vec{x}}, \vec{y})$ in $q$.

We present a reduction from $\cqa{\sjf{q}}$ to $\cqa{q}$ in \FO.

Let $\sjf{\db}$ be an instance for $\cqa{\sjf{q}}$ and $N(\alpha_1, \alpha_2, \dots, \alpha_n)$ an atom in $\sjf{q}$. Consider a mapping $\phi$ from facts to facts that for any fact $f = N(f_1, f_2, \dots, f_n)$ in $\sjf{\db}$, $$\phi(f) = \pi(N)(\langle f_1, \alpha_1 \rangle, \langle f_2, \alpha_2 \rangle, \dots, \langle f_n, \alpha_n \rangle)$$ where each $\langle f_i, \alpha_i \rangle$ is a fresh constant such that $\langle f_i, \alpha_i \rangle = \langle f_j, \alpha_j \rangle$ if and only if $f_i = f_j$ and $\alpha_i = \alpha_j$. Let $\db = \{\phi(f) \mid f \in \sjf{\db}\}$. 

We first show that $\phi$ is bijective from $\sjf{\db}$ to $\db$. By construction, $\phi$ is onto. Suppose $\phi$ is not injective, then there exist two distinct facts $f = R(f_1, f_2, \dots, f_n)$ and $g = S(g_1, g_2, \dots, g_m)$ from atoms $R(\alpha_1, \alpha_2, \dots, \alpha_n)$ and $S(\beta_1, \beta_2, \dots, \beta_m)$ in $\sjf{q}$ such that $\phi(f) = \phi(g)$. We then have $m = n$, $\pi(R) = \pi(S)$ and for each $1 \leq i \leq n$, $\langle f_i, \alpha_i \rangle = \langle g_i, \beta_i \rangle $, implying that $\pi(R)(\alpha_1, \alpha_2, \dots, \alpha_n) = \pi(S)(\beta_1, \beta_2, \dots, \beta_m)$ in $q$, but $q$ is minimal, a contradiction.

We show that $\sjf{\db}$ is a ``yes''-instance for $\cqa{\sjf{q}}$ if and only if $\db$ is a ``yes''-instance for $\cqa{q}$. Towards this end, let $\sjf{\rep}$ be a repair of $\sjf{\db}$, and consider the set $\rep = \{\phi(f) \mid f \in \sjf{\rep}\}$. It is easy to see that $\rep$ is a repair of $\db$. Hence it is sufficent to show that $\sjf{\rep}$ satisfies $\sjf{q}$ if and only if $\rep$ satisfies $q$.


Therefore, $\sjf{\rep}$ satisfies $\sjf{q}$ if and only if there exists a valuation $\mu$ such that $\mu(\sjf{q}) \subseteq \sjf{\rep}$. That is, for every fact $f = R(\mu(\alpha_1), \mu(\alpha_2), \dots, \mu(\alpha_n))$ in $\sjf{\rep}$, $\phi(f) \in \rep$. Since $\phi$ is bijective, this is equivalent to 
$\phi(\mu(q)) \subseteq \rep$. This shows that $\rep$ satisfies $q$.
The other direction follows similarly since $\phi$ is bijective.
\end{proof}

\myparagraph{Attacks}
Let $q$ be a self-join-free Boolean CQ. 
For every atom $F \in q$, we define $F^{+,q}$ as the set of all variables in $q$ that are functionally determined by $\mathsf{key}(F)$ with respect to all functional dependencies of the form $\mathsf{key}(G) \rightarrow \vars{G}$ with $G\in q\setminus\{F\}$. 
Following~\cite{KoutrisW15}, the \emph{attack graph} of $q$ is a directed graph whose vertices are the atoms of~$q$. There is a directed edge, called \emph{attack}, from $F$ to $G$ ($F\neq G$), written $F \attacks{{q}} G$, if there exists a path between $F$ and $G$ in the query such that  every two adjacent atoms share a variable not in $F^{+,q}$. The attack is called {\em weak} if 
$\mathsf{key}(F) \rightarrow \mathsf{key}(G)$, otherwise it is called {\em strong}. It was proved in~\cite{KoutrisW15} that for a self-join-free Boolean CQ $q$, $\cqa{q}$ is \coNP-hard if and only if there exist two atoms $F \neq G$ that attack each other and at least one of the attacks is strong.

\smallskip

We can now prove the proposition.

\begin{proof}[Proof of Proposition~\ref{prop:berge}]
Let $q$ be a minimal connected query in $\ibcq$.

Assume that $q$ is not a rooted tree query.
Then, there exist two atoms $R(\underline{x}, \dots, z, \dots)$ and $S(\underline{y}, \dots, z, \dots)$ with $x \neq y$ (and possibly $R = S$) in $q$. Consider now $\sjf{q}$, and let $R_0$ and $S_0$ be the corresponding atoms of $R$ and $S$ in $\sjf{q}$.

Since $q$ satisfies (1) and (3), so does $\sjf{q}$, and we have $R_0^{+,\sjf{q}} = \{x\}$ and $S_0^{+,\sjf{q}} = \{y\}$.  Hence  $R_0 \attacks{\sjf{q}} S_0$, and similarly, $S_0 \attacks{\sjf{q}} R_0$. 
By~\cite{KoutrisW15} $\cqa{\sjf{q}}$ is \LSPACE-hard (due to this cycle in the attack graph of $\sjf{q}$), and so is $\cqa{q}$ by Lemma~\ref{lemma:bridging-lemma}.

Next we additionally assume that $q$ is Berge-acyclic, that is, $q\in\aibcq$.
We thus have either  $x \not\vardetermines{q} y$ or $y \not\vardetermines{q} x$.
Indeed, otherwise there exist atoms $R_0, R_1, \dots, R_n$ and $S_0, S_1, \dots, S_m$ and variables 
$x_0, x_1,x_2,\dots, x_{n+1}$, $y_0, y_1, \dots, y_{m+1}$ in $\sjf{q}$ where $x_0 = x$, $x_{n+1} = y$, $y_0 = y$, $y_{m+1} = x$ such that $\sjf{q}$ contains atoms $R_i(\underline{x_i}, \dots, x_{i+1}, \dots)$ for every $0 \leq i \leq n$ and $S_i(\underline{y_i}, y_{i+1})$ for every $0 \leq i \leq m$. Then, 
$$\{x_0, R_0, x_1, R_1, \dots, R_{n}, x_{n+1} (=y=y_0), S_{0}, y_1, S_1, \dots, S_{m}, y_{m+1} (= x = x_0) \}$$ 
is a Berge-cycle in $\sjf{q}$, a contradiction to that $\sjf{q}$ (and $q$) are Berge-acyclic.
This implies that at least one of the two attacks is strong. 
Hence, applying the result from~\cite{KoutrisW15}, $\cqa{\sjf{q}}$ is \coNP-hard (due to this strong cycle in the attack graph of $\sjf{q}$), and so is $\cqa{q}$ by Lemma~\ref{lemma:bridging-lemma}.
\end{proof}

We define two first-order formula $\varphi_R(u)$ and $\varphi_S(u)$:
\begin{align*}
\varphi_R(u) & = \exists y R(\underline{u}, y) \land \forall y \formula{R(\underline{u}, y) \rightarrow \exists z R(\underline{y}, z)} \\
\varphi_S(u) & = \exists y S(\underline{u}, y) \land \forall y \formula{S(\underline{u}, y) \rightarrow \exists z S(\underline{y}, z)} 
\end{align*}

\begin{lemma}[\cite{DBLP:conf/pods/KoutrisOW21}]
\label{lemma:easy}
For $q = R(\underline{c}, y), R(\underline{y}, z)$, the \FO-rewriting of the query $\cqa{q}$  is $\varphi_R(c)$.
\end{lemma}

\begin{lemma}
\label{lem:ex1}
Let $q = \{R(\underline{u}, x_1), R(\underline{x_1}, x_2), S(\underline{u}, y_1), S(\underline{y_1}, y_2)\}$.
Then $\cqa{q}$ is in \FO.
\end{lemma}
\begin{proof}
Let $\db$ be a database instance.
We show that $\db$ is a ``yes''-instance for $\cqa{q}$ if and only if $\db$ satisfies the formula 
$$\exists c : \varphi_R(c) \land \varphi_S(c).$$

\framebox{$\impliedby$} This direction is straightforward. 
Let $\rep$ be any repair of $\db$. By Lemma~\ref{lemma:easy}, $\rep$ contains a path of $RR$ starting in $c$ and a path $SS$ starting in $c$. Hence $\rep$ satisfies $q$.

\framebox{$\implies$} 
Assume that $\db$ does not satisfy $\rep$. We construct a falsifying repair $\rep$ of $\db$ as follows: For each constant $c \in \adom{\db}$:
\begin{itemize}
\item if $R(\underline{c}, *)$ is empty and $S(\underline{c}, *)$ is nonempty, pick an arbitrary fact from $S(\underline{c}, *)$;
\item if $R(\underline{c}, *)$ is nonempty and $S(\underline{c}, *)$ is empty, pick an arbitrary fact from $R(\underline{c}, *)$;
\item Assume that $R(\underline{c}, *)$ and $S(\underline{c}, *)$ are nonempty. If $\varphi_R(c)$ is false,
we pick $R(\underline{c}, d_1)$ such that $R(\underline{d_1}, *)$ is empty; or otherwise $\varphi_S(c)$ is false,
we pick $S(\underline{c}, d_2)$ such that $S(\underline{d_2}, *)$ is empty.
\end{itemize} 

We argue that $\rep$ is a falsifying repair. Suppose for contradiction that $\rep$ satisfies $q$ and $u$ is mapped to $c$.
Assume that $R(\underline{c}, d_1), S(\underline{c}, d_2) \in \rep$. 
If $\varphi_R(c)$ is false, then we would have picked a fact $R(\underline{c}, d_1)$ such that $R(\underline{d_1}, *)$ is empty, a contradiction, and so is the other case where $\varphi_S(c)$ is false.
\end{proof}

\begin{lemma}
\label{lem:ex2}
Let $q = \{R(\underline{u}, x_1), R(\underline{x_1}, x_2), X(\underline{x_2}, x_3), S(\underline{u}, y_1), S(\underline{y_1}, y_2)\}$.
Then $\cqa{q}$ is in \NL-hard and in \LFPL.
\end{lemma}
\begin{proof}
The \NL-hardness proof follows by modifying the proof of Lemma 7.1 in~\cite{DBLP:conf/pods/KoutrisOW21} to also add a copy of $SS$ path starting in every vertex $v \in \{s'\} \cup V$.

Let $\db$ be a database instance. We revise the algorithm in Fig.~\ref{fig:algo} for $\cqa{RRX}$ as follows:
\begin{itemize}
\item Initialize $N = \{\pair{c}{x_2} \mid X(\underline{c}, *) \neq \emptyset\}$
\item while $N$ is not fixed:
\item \hspace{1cm} add $\pair{c}{x_1}$ to $N$ if 
$$\exists d R(\underline{c}, d) \land \forall d \formula{R(\underline{c}, d) \rightarrow f_2(R(\underline{c}, d), R(\underline{x_1}, x_2))},$$ where 
$$f_2(R(\underline{c}, d), R(\underline{x_1}, x_2)) = \pair{d}{x_2} \in N \lor \underbrace{\formula{\varphi_S(c) \land f_1(R(\underline{c}, d), R(\underline{u}, x_1))}}_{\text{rewinding to $R(\underline{u}, x_1)$ allowed only when $\varphi_S(c)$ is true}}$$
\item \hspace{1cm} add $\pair{c}{u}$ to $N$ if 
$$\varphi_S(c) \land \exists d R(\underline{c} ,d) \land \forall d \formula{R(\underline{c}, d) \rightarrow f_1(R(\underline{c}, d), R(\underline{u}, x_1))},$$ where 
$$f_1(R(\underline{c}, d), R(\underline{u}, x_1)) = \pair{d}{x_1} \in N.$$
\end{itemize}

Our algorithm essentially first computes the set $N$ and then checks $\exists c : \pair{c}{u} \in N$. 
Notice that this algorithm is in \LFPL.
To show correctness, we argue that there exists a constant $c$ such that $\pair{c}{u} \in N$ if and only if $\db$ is a ``yes''-instance for $\cqa{q}$.

To this end, we need to define a refined notion of frugal repairs that take into account of $\varphi_S(u)$ constructed as follows:
\begin{itemize}
\item pick an arbitrary fact from every nonempty block $X(\underline{c}, *)$;
\item for every nonempty block $S(\underline{c}, *)$, if $\varphi_S(c)$ is true, pick an arbitrary fact; or otherwise, pick $S(\underline{c}, d)$ such that $S(\underline{d}, *)$ is empty; and 
\item for every fact $R(\underline{c}, d)$ in the block $R(\underline{c}, *)$, we define the frugal index of $R(\underline{c}, d)$ to be 
\begin{itemize}
\item $0$, if $\varphi_S(c) \land f_1(R(\underline{c}, d), R(\underline{u}, x_1))$ is true;
\item $1$, if $\varphi_S(c) \land f_1(R(\underline{c}, d), R(\underline{u}, x_1))$ is false and $f_2(R(\underline{c}, d), R(\underline{x_1}, x_2))$ is true; or 
\item $2$, otherwise. 
\end{itemize}
We then pick the fact $R(\underline{c}, d)$ from each nonempty block $R(\underline{c}, *)$ with the largest frugal index.
\end{itemize}

\begin{definition}
An extended $RRX$-path in a repair $\rep$ is a sequence of facts $R(\underline{c_0}, c_1), R(\underline{c_1}, c_2), \dots, R(\underline{c_{n-1}}, c_{n}), X(\underline{c_n}, c_{n+1})$ in $\rep$ for some $n \geq 2$ such that for every $0 \leq i \leq n-2$, there exists an $SS$-path in $\rep$ starting in $c_i$.
\end{definition}


The following claim concludes the proof.

\begin{claim}
\label{cl:newfrugal}
The following statements are equivalent:
\begin{enumerate}
\item $\pair{c}{u} \in N$;
\item there exists an extended $RRX$-path in $\rep^*$ starting in $c$; and
\item for every repair $\rep$ in $\db$, there exists an extended $RRX$-path in $\rep$ starting in $c$.
\end{enumerate}
\end{claim}
\begin{proof}
\framebox{(3) $\implies$ (2)} Straightforward.
\framebox{(2) $\implies$ (1)} 
Let $R(\underline{c_0}, c_1), R(\underline{c_1}, c_2), \dots, R(\underline{c_{n-1}}, c_{n}), X(\underline{c_n}, c_{n+1})$ be an extended $RRX$-path in $\rep^*$ for some $n \geq 2$ and $c_0 =c$.
Since for every $0 \leq i \leq n-2$, there exists an $SS$-path in $\rep^*$ starting in $c_i$, by construction of $\rep^*$, we have that $\varphi_S(c_i)$ is true for every $0 \leq i \leq n-2$.

We have that $\pair{c_n}{x_2} \in N$. Then for $R(\underline{c_{n-1}}, *)$, we have that $f_2(R(\underline{c_{n-1}}, c_n), R(\underline{x_1}, x_2))$ is true. Then by the choice of the frugal repair $\rep^*$, the frugal index of $R(\underline{c_{n-1}}, c_n)$ is at most 1, and thus $f_2(R(\underline{c_{n-1}}, c'_n), R(\underline{x_1}, x_2))$ is true for every fact $R(\underline{c_{n-1}}, c'_n)$ in the block $R(\underline{c_{n-1}}, *)$. Hence $\pair{c_{n-1}}{x_1} \in N$.

Then, we must have $f_1(R(\underline{c_{n-2}}, c_{n-1}, R(\underline{u}, x_1))$ is true, because $\pair{c_{n-1}}{x_1} \in N$. Notice that $\varphi_S(c_{n-2})$ is true, the frugal index of $R(\underline{c_{n-2}}, c_{n-1})$ is $0$. Then by construction of $\rep^*$, for every fact $R(\underline{c_{n-2}}, c_{n-1}')$ in block $R(\underline{c_{n-2}}, *)$, $f_1(R(\underline{c_{n-2}}, c_{n-1}, R(\underline{u}, x_1))$ is true. Hence $\pair{c_{n-2}}{u} \in N$. 
Additionally, for every fact $R(\underline{c_{n-2}}, c_{n-1}')$ in block $R(\underline{c_{n-2}}, *)$, $\varphi_S(c_{n-2}) \land f_1(R(\underline{c_{n-2}}, c_{n-1}, R(\underline{u}, x_1))$ is true, and thus $f_2(R(\underline{c_{n-2}}, c'_{n-1}, R(\underline{x_1}, x_2))$ is also true. This gives $\pair{c_{n-2}}{x_1} \in N$.

This argument may continue, until we yield that $\pair{c_0}{u} = \pair{c}{u} \in N$.

\framebox{(1) $\implies$ (3)} 
Assume that $\pair{c}{u}$ in $N$.
Let $\rep$ be any repair of $\db$. We inductively construct an extended $RRX$-path in $\rep$ starting in $c$.
Assume that $R(\underline{c_0}, c_1) \in \rep$ with $c_0 = c$. 
Hence $\varphi_S(c_0)$ is true, and that $\pair{c_1}{x_1} \in N$.
Let $R(\underline{c_1}, c_2) \in \rep$, and we have $f_2(R(\underline{c_1}, c_2), R(\underline{x_1}, x_2))$ is true.
If $\pair{c_2}{x_2} \in N$, let $X(\underline{c_2}, c_3) \in \rep$ and then the proof is complete since $R(\underline{c_0}, c_1), R(\underline{c_1}, c_2), X(\underline{c_2}, c_3)$ is an extended $RRX$-path in $\rep$.
Otherwise, we have that $\varphi_S(c_1) \land f_1(R(\underline{c_1}, c_2), R(\underline{u}, x_1))$ is true. 
Therefore, $\varphi_S(c_1)$ is true and $\pair{c_2}{x_1} \in N$.
This process thus continues, until we have produced facts $R(\underline{c_0}, c_1), R(\underline{c_1}, c_2), \dots, R(\underline{c_{n-1}}, c_n)$, such that $\varphi_S(c_i)$ is true for $0 \leq i \leq n-2$ and $\pair{c_n}{x_2} \in N$, which gives a fact $X(\underline{c_n}, c_{n+1}) \in \rep$. Hence there exists an extended $RRX$-path in $\rep$ starting in $c_0 = c$.
\end{proof}

Now, assume that $\pair{c}{u} \in N$. Then by Claim~\ref{cl:newfrugal}, every repair $\rep$ of $\db$ contains an extended $RRX$-path starting in $c$, which satisfies $q$. Assume that $\pair{c}{u} \notin N$. Then by Claim~\ref{cl:newfrugal}, there is no extended $RRX$-path in $\rep^*$, and thus $\rep^*$ does not satisfy $q$. 
\end{proof}

\begin{lemma}
\label{lem:ex3}
Let $q = \{R(\underline{u}, x_1), R(\underline{x_1}, x_2), X(\underline{x_2}, x_3), S(\underline{u}, y_1), S(\underline{y_1}, y_2), Y(\underline{y_2}, y_3)\}$.
Then $\cqa{q}$ is \coNP-complete.
\end{lemma}
\begin{proof}
For \coNP-hardness, we present a reduction from \textsf{Monotone SAT}: given a monotone CNF formula $\varphi$, does $\varphi$ have a satisfying assignment?

Let $\varphi$ be a monotone CNF formula. We construct a database $\db$ for $\cqa{q}$ as follows:

\begin{itemize}
\item for each variable $z$ in $\varphi$, introduce a copy of 
$$z^- = \qcopy{\{R(\underline{u}, x_1), R(\underline{x_1}, x_2), X(\underline{x_2}, x_3), S(\underline{u}, y_1), Y(\underline{y_1},y_2)\}}{u}{z}$$
 and a copy of 
$$z^+ = \qcopy{\{R(\underline{u}, x_1), X(\underline{x_1}, x_2), S(\underline{u}, y_1), S(\underline{y_1},y_2), Y(\underline{y_2}, y_3)\}}{u}{z};$$
\item for each positive literal $z$ in a positive clause $C$, introduce a copy of 
$$C_z = \qcopy{\{R(\underline{u}, x_1), R(\underline{x_1}, x_2), X(\underline{x_2}, x_3), S(\underline{u}, y_1)\}}{u, y_1}{C, z},$$ 
\item for each negative literal $\overline{z}$ in a negative clause $\overline{C}$, introduce a copy of
$$\overline{C}_z = \qcopy{\{R(\underline{u}, x_1), S(\underline{u}, y_1), S(\underline{y_1}, y_2), Y(\underline{y_2}, y_3)\}}{u, x_1}{\overline{C}, z}.$$ 
\end{itemize}

Note that by construction, the only inconsistencies happen at each block $R(\underline{z}, *)$ and $S(\underline{z}, *)$, because each $R$-block may choose either an $R$-edge or an $RX$-edge; and each $S$-block may choose either an $S$-edge or an $SY$-edge. Surprisingly, there are $4$ possible repairs for each variable $z$, but it is sufficient to encode a Boolean choice between $z = 0$ and $z = 1$. An example gadget is shown in Fig.~\ref{fig:cfo-hardness}.

We show that $\varphi$ has a satisfying assignment if and only if there is a repair of $\db$ that does not satisfy $q$.

\framebox{$\implies$}
Assume that $\sigma$ is a satisfying assignment to $\varphi$. We now construct a repair $\rep$ of $\db$ as follows. For each variable $z$, if $\sigma(z) = 1$, then we pick $z^+ \subseteq \rep$, or otherwise we pick $z^- \subseteq \rep$; and for each positive clause $C$, there must be some literal $z \in C$ with $\sigma(z) = 1$, we pick $C_z \subseteq \rep$; and for each negative clause $\overline{C}$, there must be literal $\overline{z} \in \overline{C}$ with $\sigma(z) = 0$, we pick $\overline{C}_z \subseteq \rep$.

We argue that $\rep$ does not satisfy $q$. Indeed, by construction, for each clause $C$, we have 
$C_z \subseteq \rep$, but for the variable $z$, we picked $z^+ \subseteq \rep$, and they cannot satisfy the query $q$. Similarly, every negative clause $\overline{C}$ cannot satisfy $q$. Each variables are picked so that they also do not satisfy $q$. We conclude that $\rep$ does not satisfy $q$.

\framebox{$\impliedby$}
Assume that $\rep$ is a repair of $\db$ that does not satisfy $q$. Consider the assignment $\sigma$ that assigns $\sigma(z) = 1$, for every variable $z$ with $z^+ \subseteq \rep$, and assigns $\sigma(z) = 0$ for every $z^- \subseteq \rep$, and assign all other variables arbitrarily.

We now argue that $\sigma$ is a satisfying assignment. Let $C$ be any positive clause and assume $C_z \subseteq \rep$. Then the repair must choose the path $SSY$ starting in $z$ (or otherwise $\rep$ satisfies $q$ rooted at the clause constant $C$), and consequently the path $RX$ starting in $z$ (or otherwise $\rep$ satisfies $q$ rooted at the variable constant $z$). Hence $z^+ \subseteq \rep$, and we set $\sigma(z) = 1$. Hence every positive clause $C$ is satisfied. Let $\overline{C}$ be any positive clause and assume $\overline{C}_z \subseteq \rep$. Then the repair must similarly choose the path $RRX$ and $SY$ starting in $z$, or otherwise $\rep$ satisfies $q$. Then $z^- \subseteq \rep$, and we set $\sigma(z) = 0$. Hence every negative clause $\overline{C}$ is satisfied. 
\end{proof}

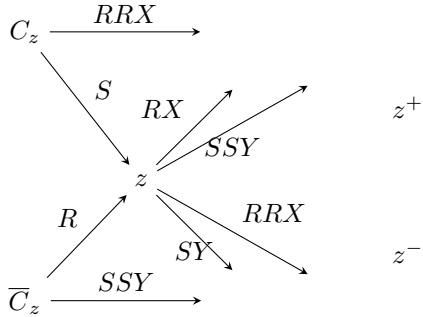
\begin{figure}[h]
\centering
\begin{tikzpicture}[->,>=stealth,auto=left, scale=0.8,vnode/.style={circle,black,inner sep=1pt,scale=1},el/.style = {inner sep=3/2pt}]
\node[] (c)  {${C}_z$};
\node[] (cend) [right=2cm of c] {};
\node[] (z) [below right=1.5cm and 1cm of c] {$z$};
\node[] (oc) [below=3cm of c] {$\overline{C}_z$};
\node[] (ocend) [right=2cm of oc] {};

\node (rx) [above right=1cm and 1cm of z] {};
\node (ssy) [above right=1cm and 2cm of z] {};
\node (rrx) [below right=1cm and 2cm of z] {};
\node (sy) [below right=1cm and 1cm of z] {};

\node (zminus) [below right=0.5cm and 3cm of z] {$z^-$};
\node (zplus) [above right=0.5cm and 3cm of z] {$z^+$};

\path[->]
(c) edge node {$RRX$} (cend)
(c) edge node {$S$} (z)
(oc) edge node {$SSY$} (ocend)
(oc) edge node {$R$} (z)
(z) edge node {$RX$} (rx)
(z) edge node {$RRX$} (rrx)
(z) edge node[below] {$SY$} (sy)
(z) edge node[below] {$SSY$} (ssy)
;
\end{tikzpicture}
\caption{The gadget used in Lemma~\ref{lem:ex3}}
\label{fig:cfo-hardness}
\end{figure}

\begin{proof}[Proof of Proposition~\ref{prop:single-pk}]
Follows from Lemma~\ref{lem:ex1},~\ref{lem:ex2} and~\ref{lem:ex3}.
\end{proof}

\end{document}